\documentclass[onecolumn, prx, superscriptaddress,notitlepage, nofootinbib]{revtex4-2}
\usepackage{graphicx}
\usepackage{dcolumn}
\usepackage{bm}
\usepackage[usenames,dvipsnames]{color}
\usepackage{amsthm}
\usepackage{enumitem} 
\usepackage{algorithm}
\usepackage{algpseudocode}

\newcolumntype{P}[1]{>{\centering\arraybackslash}p{#1}}
\usepackage{mathtools}
\usepackage[most]{tcolorbox}
\usepackage{multirow}
\usepackage{gensymb}
\usepackage[normalem]{ulem}
\usepackage{CJK}
\usepackage{comment}
\usepackage[colorlinks, linkcolor=blue,anchorcolor=blue,citecolor=blue,urlcolor=blue]{hyperref}
\usepackage{amssymb}
\usepackage{amsthm}
\usepackage{pifont}
\usepackage{physics}
\usepackage{natbib}
\newtheorem{definition}{Definition}
\newtheorem{theorem}{Theorem}
\usepackage{float} 
\makeatletter 
\let\newfloat\newfloat@ltx 
\makeatother
\usepackage{caption}
\usepackage{subcaption}
\usepackage{amsmath}
\usepackage{lineno}

\usepackage{booktabs}
\usepackage{colortbl}
\usepackage{caption}
\newcommand{\etal}{\textit{et al.}}
\newcommand{\algo}{QC-Forest}

\newtheorem{lemma}{Lemma}

\begin{document}
\title{QC-Forest: a Classical-Quantum Algorithm to Provably Speedup Retraining of Random Forest} 

\author{Romina Yalovetzky}
\email{romina.yalovetzky@jpmorgan.com}
\affiliation{Global Technology Applied Research, JPMorgan Chase, New York, NY 10017}
\author{Niraj Kumar}
\affiliation{Global Technology Applied Research, JPMorgan Chase, New York, NY 10017}
\author{Changhao Li}
\affiliation{Global Technology Applied Research, JPMorgan Chase, New York, NY 10017}
\author{Marco Pistoia}
\affiliation{Global Technology Applied Research, JPMorgan Chase, New York, NY 10017}
\date{\today}

\begin{abstract}

Random Forest (RF) is a popular tree-ensemble method for supervised learning, prized for its ease of use and flexibility. Online RF models require to account for new training data to maintain model accuracy. This is particularly important in applications where data is periodically and sequentially generated over time in data streams, such as auto-driving systems, and credit card payments. In this setting, performing periodic model retraining with the old and new data accumulated is beneficial as it fully captures possible drifts in the data distribution over time. However, this is unpractical with state-of-the-art classical algorithms for RF as they scale linearly with the accumulated number of samples. We propose QC-Forest, a classical-quantum algorithm designed to time-efficiently retrain RF models in the streaming setting for multi-class classification and regression, achieving a runtime poly-logarithmic in the total number of accumulated samples. QC-Forest leverages Des-q, a quantum algorithm for single tree construction and retraining proposed by Kumar \etal~by expanding to multi-class classification, as the original proposal was limited to binary classes, and introducing an exact classical method to replace an underlying quantum subroutine incurring a finite error, while maintaining the same poly-logarithmic dependence. Finally, we showcase that \algo~achieves competitive accuracy in comparison to state-of-the-art RF methods on widely used benchmark datasets with up to 80,000 samples, while significantly speeding up the model retrain.



\end{abstract}

\maketitle

\section{Introduction \label{section:intro}}
\label{intro}


The current era of big data accumulation brings about opportunities to advance data mining methods and tap into data patterns at a scale  which was not possible before \cite{sagiroglu2013big}. However, this also comes at the cost of the models requiring massive computational resources to train given the large training data used. A widely popular class of data mining models in the context of supervised learning is decision tree-based models, including random forest (RF) \cite{breiman2001random} and boosting algorithms \cite{chen2015xgboost, solomatine2004adaboost, drucker1995boosting}. These models are defined in the context of ensemble learning for classification and regression by combining the output of multiple independent trees, also called tree-bagging. While conceptually simple, they have consistently provided competitive performance across a variety of supervised learning tasks \cite{parmar2019review}. 


Despite its widespread success, constructing each tree in the ensemble is computationally resource-intensive, especially for large datasets. State-of-the-art techniques for constructing each tree are based on threshold feature splitting \cite{breiman2017classification,Lewis2000AnIT, quinlan1986induction, quinlan2014c4}, whose algorithmic complexity scales polynomially with the number of training samples. Once the model is trained and put \textit{online}, it needs to be updated to account for new training data in order to maintain model accuracy. This is particularly important in applications where data is periodically and sequentially generated over time in the form of data streams \cite{lu2018learning, zhou2023ods, wang2022learngene}. Some practically-relevant scenarios with data streams include auto-driving systems, weather forecasts, and credit card payments. In this scenario, machine learning (ML) models are susceptible to degradation in the performance due to drift in the concept or distribution of the dataset over time \cite{tsymbal2004problem}, requiring regularly update to account for the new information.  Consequently, \emph{static} models trained solely on historical data may fail to provide accurate prediction on new test samples \cite{li2022ddg}. 

In the literature of classical ML, several techniques have been developed to effectively account for new information. These include \emph{data subsampling techniques} \cite{shaw1990distributed}, where the model is trained over subsamples of data; \emph{incremental learning} \cite{domingos2000mining, manapragada2018extremely}, which is an online learning algorithm to update the model parameters using only the new training data. Although these models take new information into consideration, they usually suffer from concept drift in data streams as they limit to considering only subsets of data. The most effective way of addressing this problem is with \emph{periodic retraining with incremental batches of data}, where the model is regularly retrained with a combination of old and new data. This approach is particularly beneficial as it fully captures the drift in the data distribution over time and prevents the model from performance degradation over time \cite{domingos2000mining}. However, frequent retraining requires periodic data accumulation and as the retraining of the models scales linearly on the number of the accumulated samples, it quickly becomes time-consuming and therefore unpractical.

Paralleling development in classical ML, the emergence of quantum computing has opened up the opportunity to speedup many computational tasks. In particular, ML has been identified as an area that could benefit significantly from quantum technology, with multiple quantum algorithms having been proposed to speed up ML tasks, such as classification and distribution learning \cite{harrow2009quantum,lloyd2013quantum,kerenidis2016quantum, chakrabarti2022universal, montanaro2016quantum}. A sub-area of quantum ML algorithms are the variational algorithms, which have been identified as potential candidates for solving practically-relevant use cases in ML on the so-called noisy intermediate-scale quantum (NISQ) devices \cite{preskill2018quantum}. Despite the widespread efforts, many challenges have been identified for prospects of quantum advantage in the ML space. To name a few, variational algorithms suffer from barren plateaus \cite{larocca2024review, fontana2023adjoint, mcclean2018barren, cerezo2021cost, wang2021noise, martin2023barren} constraining their practical usage. It has also been conjectured and widely believed that the commonly used quantum variational models with a provable absence of barren plateaus could also be classically simulatable \cite{cerezo2023does}. Another challenge is that some quantum algorithms for ML require access to quantum random access memory (QRAM) capabilities \cite{giovannetti2008quantum}. While QRAM would allow for
quantum query of the data in logarithmic time, the proposed quantum-accessible classical data structures to construct QRAM from scratch would incur a time polynomial in the data size. Thus 
exponential speedup with quantum algorithms is not achievable if one accounts for the time for the creation of this specialized data structure \cite{lloyd2014quantum, kerenidis2016quantum, kerenidis2019q}. 
Furthermore, some quantum algorithms provide a speedup only when the task is to recover the output quantum state or to extract few properties of the state such as its overlap with a target state. A common example is the well-known HHL algorithm for solving linear systems of equations \cite{harrow2009quantum}. However, if the task is to recover the entire classical representation of the state, this requires performing a quantum state tomography which eliminates any speedup offered by the original algorithm \cite{aaronson2015read}.

These aforementioned challenges have left open the question of how quantum computing can concretely address needs in the ML domain and how the community should think about quantum advantage in ML \cite{cerezo2023does, cerezo2022challenges, schuld2022quantum}. Aiming for a true end-to-end theoretical speedup over classical ML methods, we introduce \algo, a hybrid classical-quantum algorithm designed to construct and time-efficiently retrain RF models with periodic incremental batches of data. Previous quantum algorithms using Grover's search \cite{grover1996fast} have targeted the speedup of the underlying tree construction in RF models, offering a quadratic improvement based on the number of features ($d$) in the dataset \cite{grover1996fast, Lu2013QIP, Khadiev2019arxiv, Heese2022Quantum}. However, in big data contexts and with periodic retraining with accumulated data, the primary bottleneck is typically the number of training examples, $N$, which can be of the order of billions. Therefore, a meaningful speedup should focus on reducing the complexity relative to $N$. 

\subsection{Novel Contributions}
QC-Forest is designed to tackle the  aforementioned issues by constructing and subsequently retraining the underlying trees in the ensemble with a runtime scaling logarithmically with $N$. Our algorithm leverages the methodology of Des-q \cite{kumar2023desq} proposed to construct and retrain single trees. We elaborate on top of Des-q to build the underlying trees in the ensemble. Des-q uses a supervised version of the quantum $k$-means algorithm \cite{kerenidis2019q} to construct and retrain individual trees, employing clustering to create child nodes from the root to the maximum desired depth. This supervised clustering uses a weighted distance in the $k$-means objective function, where the weights are normalized relevancy coefficients of each feature vector with the target label vector, efficiently estimated on a quantum computer. Des-q requires a quantum superposition query to the data elements, necessitating a quantum-accessible classical data structure like the KP-tree \cite{kerenidis2016quantum} to query data elements in superposition. Creating this data structure from scratch is time-consuming, scaling linear in $N$. This is precisely the reason of why Des-q, and also our QC-Forest algorithm, take time linear in $N$ for the initial model construction. Nevertheless, once the data has been loaded into the data structure and the model is constructed and put online, retraining with the old and new small batch of data $N_{\text{new}}$, which means training with $N + N_{\text{new}}$ data samples, is exponentially faster in comparison to classical standard methods, assuming $N_{\text{new}} \ll N$. This efficiency results from the fact that updating the quantum-accessible data structure takes time linear in $N_{\text{new}}$, while running the algorithm itself takes time logarithmic in $N$. This assumption is realistic in many applications in the industrial setting where data is periodically and sequentially generated over time in data streams. This is the case of auto-driving systems and credit card payments where the models are trained and make online and as time passes, and labelled data accumulates, periodic retraining with the accumulated data is possible. 

We leverage Des-q to create an ensemble RF model that is suitable for construction and time-efficient retraining for a variety of tasks: binary and multi-class classification and regression. For this, we make key improvements in the original algorithm to overcome its limitations and address unanswered questions. The enhancements in QC-Forest include:

\begin{enumerate}
    \item We extend Des-q to tackle multi-class classification problems by successfully tackling the following challenges: (a) we efficiently estimate the feature weights beyond the binary classification case. Des-q utilizes the Pearson correlation, limiting to only binary classes. In this work we introduce a classical method to efficiently calculate $\eta$ coefficient \cite{levine2002eta}, which handles multi classes and (b) we efficiently estimate the probability of each of the multi classes in the leaf nodes of the trees on the quantum computer. Des-q does not estimate the probability of each of the two classes, it outputs the mean value among the labels. In the case of binary classification, this corresponds to the majority label as it assumes a threshold of $0.5$. This is unsuitable for datasets with two classes that are unbalanced and for multi-class classification. In this work we introduce a novel method to efficiently estimate the probabilities of each of the multi classes. 
    \item We introduce efficient classical methods to compute the feature weights exactly for $N + N_{\text{new}}$ samples, leveraging prior feature weight information from $N$ samples. Feature weights are crucial in the supervised clustering procedure during retraining of the underlying trees.  Our method takes time linear in  $N_{\text{new}}$ and is independent of $N$. In contrast, Des-q required a quantum computer to estimate feature weights with bounded error, adding logarithmic complexity in $N$ in addition to linear dependence in $N_{\text{new}}$.
    \item Introducing an algorithm in the setting of ensemble of trees, particularly random forest, that we name QC-Forest. We show with numerical simulations that the quality of the predictions are on par with classical standard methods. We utilize the classical implementation of the algorithm and perform extensive numerical analysis over different well-known datasets with different sizes up to $80000$ samples and multiple tasks: binary and multi-class classification and regression.

\end{enumerate}

The remaining of the article is organized as follows. We introduce the algorithmic components of \algo~by separating in two tasks: BUILD and RETRAIN respectively. In Section \ref{QC-Forest-BUILD}, we introduce \algo-BUILD that builds the RF model for the first time and Section \ref{QC-Forest-RETRAIN} elaborates on how the algorithm is used for retraining the online RF model to account for new data, this is named \algo-RETRAIN. Then, in Section \ref{algorithmic_ingredients} it is explained in detail the algorithmic ingredients of \algo. Section \ref{section:numerics} discusses the classical numerical simulations performed to assess the quality of the predictions expected with \algo~and finally Section \ref{section:conclusion} discusses the main results and concludes the article.


\section{QC-Forest-BUILD: Initial Construction of RF Model  \label{QC-Forest-BUILD}}

\begin{algorithm}[H]
   \caption{QC-Forest-BUILD} \label{alg:build_algo}
\begin{algorithmic}
   \State {\bfseries Input:} Pre-processed data $D_{N}=\{(\boldsymbol{x}_i, y_i)\}_{i=1}^{N}$ where each $(\boldsymbol{x}_i, y_i) \in \mathbb{R}^{d} \times \mathcal{M}$, with $\mathcal{M} = \{l_1, \dots, l_{|\mathcal{M}|}\}$; $T$: the number of decision trees in the ensemble; $D$: the maximum depth of the trees, $k$: the number of clusters during each clustering step, and $K$: maximum number of clustering iterations.  

   \State {\bfseries Output:} For each $\text{tree}_n$, $n \in [T]$, the classical tree structure: the centroids $\{\boldsymbol{c}_{\text{node}} \in \mathbb{R}^{d}\}_{_{\text{node}}}$ corresponding to each internal node, and the leaf labels $\{label_{m}\}_{m \in [k^D]}$; the feature weight vectors $\{\boldsymbol{w}(\mathcal{S}_{n})\}_{n \in [T]}$ stored in classical memory; $\{\mathcal{S}_{n}\}_{n \in [T]}$ stored in quantum-accessible data structure.
    \For{$n=1\cdots T$} 
   \State {\bfseries Step 1 (a):} Randomly sample with replacement $\mathcal{S}_{n} \subset D_N$ training examples to be used to build the $n$-th tree $\text{tree}_n$. 
   \State {\bfseries Step 1 (b):} Classically calculate the feature weights using $\mathcal{S}_{n}$: $\boldsymbol{w}(\mathcal{S}_{n}) := \{w_j(\mathcal{S}_{n})\}_{j=1}^d$. 
   \State {\bfseries Step 1 (c):} Classically calculate the weighted centroids $\{\boldsymbol{c}_{lw}\}_{l=1}^{k}$ by doing element-wise multiplication between the centroids and the weights.
   \State {\bfseries Step 2:} Load $\mathcal{S}_{n}$ and the weighted centroids $\{\boldsymbol{c}_{lw}\}_{l=1}^{k}$ to the quantum-accessible data structure such that the data contents can be retrieved either classically or as quantum superposition. Store the feature weights $\boldsymbol{w}(\mathcal{S}_{n})$ in classical memory. 
   \State {\bfseries Step 3:} Construct the $\text{tree}_n$ by accessing the contents of $\mathcal{S}_{n}$ and $\{\boldsymbol{c}_{lw}\}_{l=1}^{k}$ as quantum-amplitude encoded states and perform the repeated supervised $q$-means to split each internal node into its children, which are the nodes of the next depth. The weighted centroids (or centroids themselves) $\{\boldsymbol{c}_{\text{node}}\}_{\text{node}}$ corresponding to the children nodes are stored in classical memory. 
   \State {\bfseries Step 4:} Repeat \textbf{Step 3} to grow the $\text{tree}_n$ up to the maximum depth $D$. 
   \State {\bfseries Step 5:} Classically estimate the label information of each leaf nodes of $\text{tree}_n$: $\{label_{m}\}_{m \in [k^D]}$. This corresponds to the probability of each class ($label_{m} = \{P_{i}\}_{i=0}^{|\mathcal{M}|}$) for classification and the mean value for regression. 
   \EndFor
\end{algorithmic}
\end{algorithm}

We introduce QC-Forest-BUILD which constructs the initial RF model. Consider an initial dataset with $N$ numerical training examples: $D_{N}=\{(\boldsymbol{x}_i, y_i)\}_{i=1}^{N}$ where each $(\boldsymbol{x}_i, y_i) \in \mathbb{R}^{d} \times \mathcal{M}$, with $\mathcal{M} = \{l_1, \dots, l_c\}$. If some features are categorical, encoding techniques, such as hot-encoding, are utilized. The algorithm diagram is shown in Fig.\ref{fig:diagram} and the algorithm in pseudo-code in Algorithm \ref{alg:build_algo}. We elaborate on each step to construct each $n$-th tree in the ensemble below making reference to the steps shown in the figure. The construction of each tree can be done either in parallel or sequentially. Note that the construction and retraining of RF models are performed with QC-Forest and the evaluation of the models is done on the classical computer. The output of QC-Forest is the structure of the trees in the ensemble ready to be used for evaluation: the centroids $\{\boldsymbol{c}_{\text{node}} \in \mathbb{R}^{d}\}_{\text{node}}$ corresponding to each internal node, and the leaf labels $\{label_{j}\}_{j \in [k^D]}$. At this point, for the algorithmic complexities, we avoid reporting the scaling with constants and estimation errors but they are reported later in Theorem \ref{theo:complexity_retrain} and Section \ref{algorithmic_ingredients}. 

\begin{itemize}
    \item  \textbf{Step 1 [classical calculations]}: A subset of data $\mathcal{S}_{n} \subset D_N$ is sampled in \textbf{(a)} by doing sampling with replacement (it can also incorporate feature bagging). In \textbf{(b)},  the feature weights $\boldsymbol{w}(\mathcal{S}_n) := \{w_j(\mathcal{S}_n)\}_{j=1}^d$ are calculated exactly in time $\mathcal{O}(Nd)$ for regression and $\mathcal{O}(Nd|\mathcal{M}|)$ for multi-class classification. We set up the initial $k$ centroids, using techniques such as k-means++ \cite{arthur2007k}, and we compute the \textit{weighted centroids} in \textbf{(c)} by doing element-wise multiplication between the centroids and the weights. Whereas step \textbf{(a)} and \textbf{(b)} are performed once, step \textbf{(c)} is performed in each iteration of supervised $q$-means.
    
    \item \textbf{Step 2 [storing data]}: The samples $\mathcal{S}_{n}$ are loaded to the quantum-accessible data structures, the \emph{KP-tree} \cite{kerenidis2016quantum} (blue disk in the figure) taking time $\mathcal{O}(Nd\log^2(Nd))$ and once constructed, the data is queried in quantum superposition in time $\mathcal{O}(\emph{poly} \log(Nd))$. These samples and the feature weights are also stored in a classical data structure on the classical computer (black disks in the figure) as we will use them to do the incremental retraining later. The weighted centroids are also loaded to a \emph{KP-tree} structure, and they are queried in $\mathcal{O}(\emph{poly} \log(kd))$. 
    
    \item \textbf{Step 3 [weighted clustering]}: The samples $\mathcal{S}_{n}$ and the weighted centroids are queried as quantum states in superposition and with them, the supervised quantum $k$-means, also called $q$-means, is performed. The weighted distance between sample $\boldsymbol{x}_i$ and cluster centroids $\boldsymbol{c}_{l}$ is defined by $\|\boldsymbol{x}_i - \boldsymbol{c}_{l}\|_w = \sqrt{\sum_{j=1}^d w_j \cdot (x_{ij} - c_{lj})^2}$. This is estimated by doing \textit{weighted inner product}, $I_w$, between the samples and centroids. This is equivalent to doing the inner product between the \textit{weighted centroids} defined as $\boldsymbol{c}_{lw} := \boldsymbol{w} \circ \boldsymbol{c}_{l} = [w_1c_{l1}, \cdots w_dc_{ld}]^T, \hspace{2mm} l \in [k] $ and the samples: $I_w(\boldsymbol{x}_i, \boldsymbol{c}_{l}) = \sum_{j=1}^d w_i x_{ij}c_{lj} = \boldsymbol{x}_i \cdot \boldsymbol{c}_{lw} = I(\boldsymbol{x}_i, \boldsymbol{c}_{lw})$. The total time complexity is $\mathcal{O}(\emph{poly} \log(Nd) KDk^{3D}d)$, where $k$ is the number of clusters, $K$ is a fixed maximum number of iteration steps and $D$ is the maximum depth of the tree. After convergence, the weighted centroids (or the centroids themselves) are stored in the classical data structure. These are the centroids corresponding to each internal node $\boldsymbol{c}_{\text{node}}$ that will be used for evaluation.  
    In general, clustering is applied to unlabelled data, whereas decision-tree models are supervised methods aiming to minimize the label class impurity in classification or label variance for regression. Simply performing clustering might not result in achieving this reduction. As we will show with the numerical simulations (Section \ref{section:numerics}), with the proposed weighted/supervised clustering, QC-Forest reduces the impurity on par with standard methods. 
    
    \item \textbf{Step 4 [tree grow]}: Clustering is sequentially repeated to grow the tree until reaching the maximum tree depth $D$, taking total time $\mathcal{O}(\emph{poly} \log(Nd) K D k^{3D}d)$.
    
    \item \textbf{Step 5 [label extraction]}: Once the tree is built, the leaf label, $label_{j}$, for each $j$-the leaf node, has to be extracted. This corresponds to the probability of each $i$-th class in classification ($\{P_{i}\}_{i}$) and the mean value in regression. Given a leaf node with the assigned training samples as quantum states, we utilize a proposed novel method to efficiently extract this information, in time $\mathcal{O}(\emph{poly} \log(Nd))$.
\end{itemize}

\begin{figure}[h!]
    \centering
    \includegraphics[width=10cm]{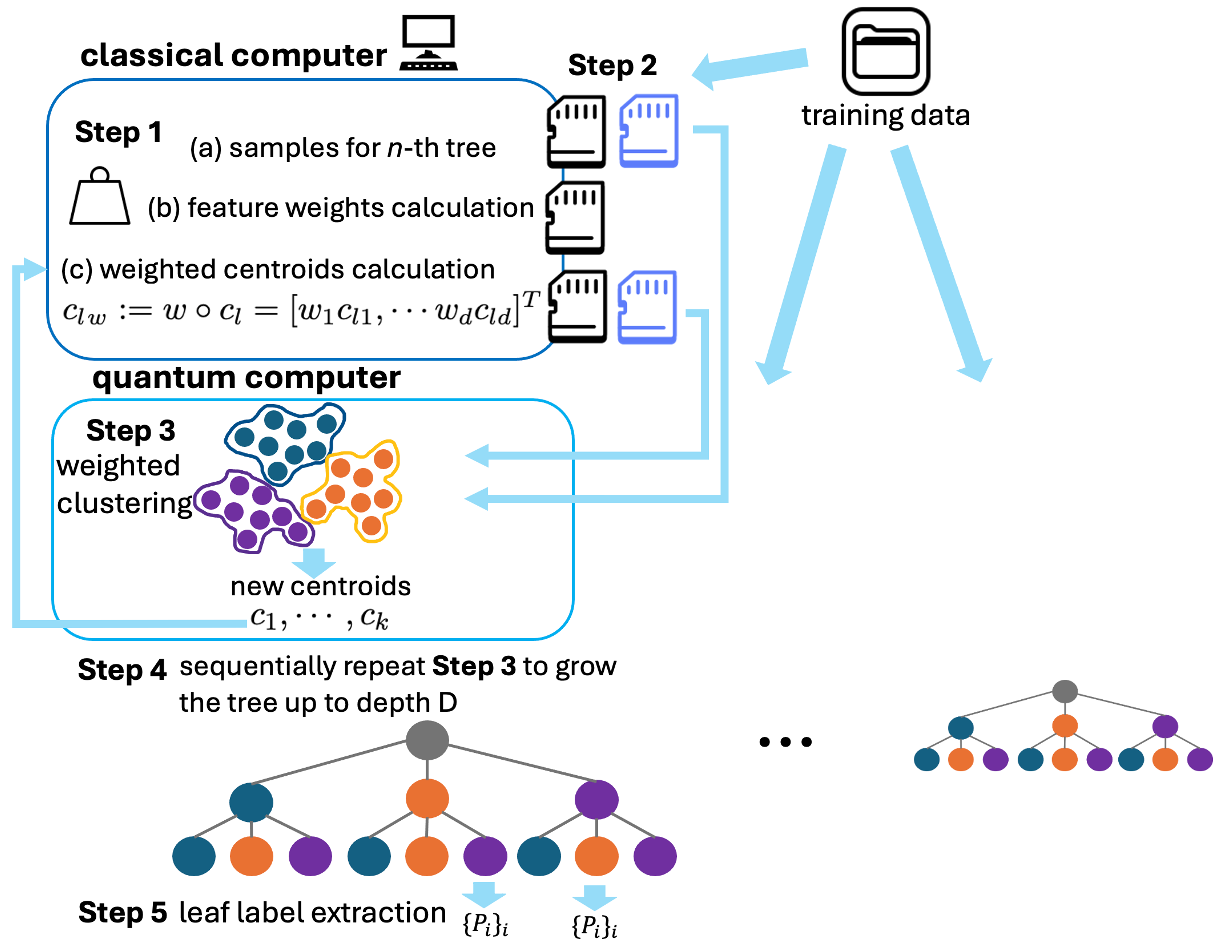}
    \caption{Diagram of QC-Forest-BUILD. For each $n$-th tree, \textbf{Step 1} consists of sampling from the training data to get $\mathcal{S}_{n}$ and calculating both the feature weights and the weighted centroids ($\{\boldsymbol{c}_{lw}\}_{l}$) on the classical computer. In \textbf{Step 2} the data is stored, and both $\mathcal{S}_{n}$ and $\{\boldsymbol{c}_{lw}\}_{l}$ are loaded into a quantum-accessible data structure (blue disk). The samples $\mathcal{S}_{n}$ and its corresponding feature weights are stored in the classical data structure (black disk) as they will be used when doing retraining with this data and a new batch of data. In \textbf{Step 3}, the supervised $q$-means is performed, where in each iteration some new centroids are output and utilized to create the weighted centroids to be used in the next iteration. Once it converges, the weighted centroids (or the centroids themselves) are stored in classical memory. This is repeated until reaching the maximum depth $D$ in \textbf{Step 4}. In \textbf{Step 5}, the leaf label extraction is performed. The probability of each class ($\{P_{i}\}_{i}$) for classification (or the mean value for regression) is estimated.}
    \label{fig:diagram}
\end{figure}

There are two components that are time-consuming in QC-Forest-BUILD, scaling $\mathcal{O}(\emph{poly}(N))$, whereas the rest of the steps scale $\mathcal{O}(\emph{poly} \log(N))$. These components are the loading of $\mathcal{S}_{n}$ to the quantum-accessible data structure (\textbf{Step 2}), and the feature weight calculation (\textbf{Step 1 (b)}). We will show how these bottlenecks are overcome when doing incremental retraining with QC-Forest-RETRAIN in the next section.

\section{QC-Forest-RETRAIN: Incremental Retraining of RF Model \label{QC-Forest-RETRAIN}}

We first define the concept of retraining with incremental batches of data as below, 

\begin{definition} \label{def:incremental_batch} \textbf{Retraining with Incremental Batches of Data}. 
 Assume we have access to periodic new batches of labeled data of sizes $\{N_{\text{new}, t}\}_{t=1}^{\mathcal{T}}$, where $t$ corresponds to time units when the new batch of data of size $N_{\text{new}, t}$  is obtained. Here we assume that each batch size $N_{\text{new}, t} \ll N$. Retraining the model with incremental batches of data at time $\mathcal{T}$ refers to training a given model with the training data of size $N + \sum _{t=1}^{\mathcal{T}} N_{\text{new}, t}$.
\end{definition}

A practical example of incremental retraining with new data batches is detecting fraudulent transactions in credit card payments. Due to the high volume of daily transactions and the frequent manual review of a subset of these transactions, labeled examples are continuously generated. This ongoing stream of new data can be used to regularly retrain the model, enhancing its accuracy and responsiveness to emerging fraud patterns.

Without loss of generality, we consider the case of $t=1$ where we have a model trained with $N$ training examples with QC-Forest-BUILD and our objective is to retrain the model with $N + N_{\text{new}}$ examples. In order to achieve this, we propose the algorithm QC-Forest-RETRAIN as highlighted in Algorithm~\ref{alg:main_algo}.

\begin{algorithm}[H]
   \caption{QC-Forest-RETRAIN}
   \label{alg:main_algo}
\begin{algorithmic}
   \State {\bfseries Input:} Pre-processed data $D_{N_{\text{new}}}=\{(\boldsymbol{x}_i, y_i)\}_{i=1}^{N_{\text{new}}}$ where each $(\boldsymbol{x}_i, y_i) \in \mathbb{R}^{d} \times \mathcal{M}$, with $\mathcal{M} = \{l_1, \dots, l_{|\mathcal{M}|}\}$; Output of Algorithm~\ref{alg:build_algo}; $T$: the number of decision trees in the ensemble; $D$: the maximum depth of the trees, $k$: the number of clusters during each clustering step, and $K$: maximum number of clustering iterations. 
   \State {\bfseries Output :} For each $\text{tree}_n$, $n \in [T]$, the classical tree structure: the centroids $\{\boldsymbol{c}_{\text{node}} \in \mathbb{R}^{d}\}_{\text{node}}$ corresponding to each internal node, and the leaf labels $\{label_{m}\}_{m \in [k^D]}$; feature weight vectors $\{\boldsymbol{w}(\mathcal{S}_{n} + \mathcal{\Tilde{S}}_{n})\}_{n \in [T]}$; $\{\mathcal{S}_{n} \cup \mathcal{\Tilde{S}}_{n}\}_{n \in [T]}$ stored in quantum-accessible data structure.

    \For{$n=1\cdots T$} 
   \State {\bfseries Step 1 (a):} Randomly sample with replacement $\mathcal{\Tilde{S}}_{n} \subset D_{N_{\text{new}}}$ training examples to be used to retrain the the $n$-th tree, $\text{tree}_n$. 
   \State {\bfseries Step 1 (b):} Classically update the feature weight value $\boldsymbol{w}(\mathcal{S}_{n} + \mathcal{\Tilde{S}}_{n}) = \{w_j(\mathcal{S}_{n} + \mathcal{\Tilde{S}}_{n})\}_{j=1}^d$ given $\boldsymbol{w}(\mathcal{S}_{n})$. 
   \State {\bfseries Step 1 (c):} Classically calculate the weighted centroids $\{\boldsymbol{c}_{lw}\}_{l=1}^{k}$ by doing element-wise multiplication between the centroids and the weights. 
   \State {\bfseries Step 2:} Load $\mathcal{\Tilde{S}}_{n}$ and the weighted centroids $\{\boldsymbol{c}_{lw}\}_{l=1}^{k}$ to the quantum-accessible data structure such that the data contents can be retrieved either classically or as quantum superposition. Update the stored feature weights in classical memory by $W(\mathcal{S}_{n} + \mathcal{\Tilde{S}}_{n})$.
   \State {\bfseries Step 3:} Construct the $\text{tree}_n$ by accessing the contents of $\mathcal{S}_{n} \cup \mathcal{\Tilde{S}}_{n}$ and $\{\boldsymbol{c}_{lw}\}_{l=1}^{k}$ as quantum-amplitude encoded states and perform the repeated supervised $q$-means to split each internal node into its children, which are the nodes of the next depth. The weighted centroids (or centroids itself) $\{\boldsymbol{c}_{\text{node}}\}_{\text{node}}$ corresponding to the children nodes are stored in classical memory.
   \State {\bfseries Step 4:} Repeat \textbf{Step 3} to grow the $\text{tree}_n$ up to the maximum depth $D$.
   \State {\bfseries Step 5:} Classically estimate the label information of each leaf nodes of $\text{tree}_n$: $\{label_{m}\}_{m \in [k^D]}$. This corresponds to the probability of each class ($label_{m} = \{P_{i}\}_{i=0}^{|\mathcal{M}|}$) for classification and the mean value for regression. 
   \EndFor
\end{algorithmic}
\end{algorithm}

Once the RF model is online, for each $n$-th tree in the forest, QC-Forest-RETRAIN consists in loading the set of new training examples $\mathcal{\Tilde{S}}_{n}$, sampled from the $N_{\text{new}}$ samples into the \emph{KP-tree} data structure. Since the previous $\mathcal{S}_{n}$ samples have already been loaded when running QC-Forest-BUILD, it is only required to load $\mathcal{\Tilde{S}}_{n}$, taking time $\mathcal{O}(N_{\text{new}}d\log^2(N_{\text{new}}d))$. Under the assumption of $N_{\text{new}} \ll N$, this procedure is effectively independent of $N$. The old and new data contents can be accessed in quantum superposition in time $\mathcal{O}(\emph{poly} \log((N + N_{\text{new}})d)) \approx \mathcal{O}(\emph{poly} \log(Nd))$. 

Another time-intensive step is the feature weight calculation. In this work, we introduce the $\eta$ coefficient that handles multi-classes to be used to calculate the feature weights. Classically computing this coefficient or the Pearson correlation for $\mathcal{S}_{n} + \mathcal{\Tilde{S}}_{n}$ training examples for $\text{tree}_n$ would require time $\mathcal{O}((\mathcal{S}_{n} + \mathcal{\Tilde{S}}_{n})d)\approx \mathcal{O}(Nd)$. We leverage the access to the weights calculated with $\mathcal{S}_n$ samples, $\boldsymbol{w}(\mathcal{S}_n)=\{ w_{j}(\mathcal{S}_n)\}_{j \in [d]}$, that we stored while running QC-Forest-BUILD (\textbf{Step 2}), to update these weights using $\mathcal{\Tilde{S}}_{n}$ examples and obtain $\boldsymbol{w}(\mathcal{S}_{n} + \mathcal{\Tilde{S}}_{n})$. This is efficiently done by leveraging the linearity of these two quantities, taking time  $\mathcal{O}(N_{\text{new}}d)$ and hence becomes significantly time-efficient. We discuss the details of this procedure in Section~\ref{sec:feature_weight}.

The remaining steps of supervised $q$-means followed by tree growth remain the same as the QC-Forest-BUILD with the difference being that each tree in the ensemble, $\text{tree}_n$, is built using $\mathcal{S}_{n} + \mathcal{\Tilde{S}}_{n}$ training examples in contrast to sampling the examples solely from $\mathcal{S}_n$ as it is done in QC-Forest-BUILD procedure. In the theorem below, we demonstrate that QC-Forest-RETRAIN runs in time poly-logarithmically in $N$ for retraining, thus ensuring an efficient retrain procedure for the RF model. 

\begin{theorem}[Time complexity of QC-Forest-RETRAIN]
\label{theo:complexity_retrain}
    Given the new training data $D_{N_{\text{new}}}=\{(\boldsymbol{x}_i, y_i)\}_{i=1}^{N_{\text{new}}}$ where each $(\boldsymbol{x}_i, y_i) \in \mathbb{R}^{d} \times \mathcal{M}$, with $\mathcal{M} = \{l_1, \dots, l_c\}$, such that $N_{\text{new}}\ll N$, and the output of QC-Forest-BUILD (Algorithm~\ref{alg:build_algo}), there is a a hybrid quantum-classical algorithm, QC-Forest-RETRAIN, to retrain the ensemble of $T$ trees of maximum depth $D$ with $N+N_{\text{new}}$ examples in time,
    \begin{equation}
    \begin{split}
        T_{\text{retrain}} &\approx T_{\text{load-new}} +  T_{\text{weights-update}} +   T_{\text{clustering}} + T_{\text{leaf-label}} \\ 
        &\approx T \cdot \mathcal{O}( poly \log(Nd))
    \end{split}
    \end{equation}
       where, $T_{\text{load-new}}$ is the time to load the new samples in the KP-tree, $ T_{\text{weights-update}}$ the time to update the feature weights classically, $T_{\text{clustering}}$ is the time to perform repeated supervised $q$-means and grow the tree up to maximum depth $D$, and $T_{\text{leaf-label}}$ is the extraction of output label information from leaf nodes.
\end{theorem}
\begin{proof}
We analyze the complexity of each component: \\
1. \emph{Loading new examples ($T_{\text{load-new}}$)}: As highlighted in Appendix~\ref{appendix:data_loading}, the time to load the new examples depends on the size of the new training data, i.e.,
$T_{\text{load-new}} = \mathcal{O}(N_{\text{new}}d\log^2(N_{\text{new}}d))$ using the \emph{KP-tree} structure. Given $N_{\text{new}}\ll N$, this is independent of $N$. 

\noindent 2. \emph{Classical update the feature weights ($T_{\text{weights-update}}$)}: $T_{\text{weights-update}} =  \mathcal{O}(N_{\text{new}}d|\mathcal{M}|)$ when using $\eta$ coefficient or $ \mathcal{O}(N_{\text{new}}d)$ when using Pearson correlation. This is the time taken to update the features weights and load them in the KP-tree. Refer to Theorem \ref{thm:pearson_correlation} and \ref{thm:eta_coef} for the proofs.  
  
3. \emph{Quantum clustering and tree grow ($T_{\text{clustering}}$)}: As highlighted in Section \ref{sec:weighted_clustering} using the results of Des-q algorithm by Kumar \etal \cite{kumar2023desq},
    \begin{equation}
       T_{\text{clustering}} \approx  \mathcal{O}(\emph{poly} \log(Nd) C_1)
    \end{equation}
    where $C_1 = \frac{K D k^{3D}d \log^2 (k)\log^2(p)\log^2(1/\Delta)\eta_1^2\eta_2}{\epsilon_1^2\epsilon_2}$, $K$ is the maximum iterations in the supervised $q$-means clustering, $\epsilon_1$ is, by a multiplicative factor, the error in the estimation of the inner product between the training examples and the centroids, estimated with probability $1-2\Delta$ as part of the clustering, $\eta_1 = \text{max}_{i \in [N]}(\|\boldsymbol{x}_i\|^2)$, $\epsilon_2$ is, by a multiplicative factor, the estimation error of the centroids of the tree's internal nodes, part of the centroid assignment during clustering, $\eta_2 = \text{max}_{j \in [d]}\|x_{j}\|$ and $p$ is the number of qubits to encode the distance between the weighted centroids and training examples. 
    
4. \emph{Output label information ($T_{\text{leaf-label}}$)}: For the task of regression, we show in Section~\ref{sec:leaf_assign} that extracting the label value for each leaf node within a precision $\epsilon_3$ takes time,
\begin{equation}
       T_{\text{leaf-label}} \approx  \mathcal{O}(T \emph{poly} \log(Nd) C_2)
    \end{equation}
    where $C_2 = \frac{ Dk^{3D}d \log (k)\log(p)\log(1/\Delta)\eta_1\eta_3}{\epsilon_1 \epsilon_3}$ and $\eta_3 = \|Y\|$. 
    
For classification, the relative occurrence of each class is computed: ${p_{jl}} = \frac{N_{jl}}{\sum_{l \in [\mathcal{M}]}N_{jl}}$, where $j\in[k^D],l\in[\mathcal{M}]$ with $N_{jl}$ being the number of examples in the $j$-th cluster with label $l$. As proved in Section~\ref{sec:leaf_assign}, the classical estimation of $p_{jl}$ for the leaf nodes and each label $l$ takes time,
\begin{equation}
       T_{\text{leaf-label}} \approx  \mathcal{O}(T \emph{poly} \log(Nd) C_3)
    \end{equation}
    where $C_3 = C_2 \cdot  \mathcal{O}(|\mathcal{M}|\log(|\mathcal{C}_{j,D}|\log (|\mathcal{M}|))\log(|\mathcal{M}|))$ and $|\mathcal{C}_{j,D}|$ corresponds to the training samples assigned to the $j$-th cluster/node at depth $D$, $|\mathcal{C}_{j,D}| \approx \mathcal{O}(N)$. 
\end{proof}

It becomes evident that the proposed algorithm can become memory-intensive, as it requires storing training samples in the data structure that allows retrieval either classically or as a quantum superposition during incremental retraining with new data batches. This means the stored data size will continually increase. To manage this, old data could be periodically removed based on criteria such as its age, or replaced with reduced data that retains the most significant information, such as through principal component analysis.

While the total samples used for training are incrementally stored as $N + \sum _{t=1}^{\mathcal{T}} N_{\text{new}, t}$ over different $\mathcal{T}$, increasing the memory complexity, the feature weights maintain fixed memory complexity as they are updated. Additionally, the weighted centroids are continually overwritten in both classical and quantum-accessible memory during each tree's retraining and at each clustering iteration.

Now that we have introduced both the construction and the retraining with QC-Forest, we can comment on how the inference with the RF models created by the proposed algorithm takes place. Both versions of the proposed algorithm, QC-Forest-BUILD, and QC-Forest-RETRAIN, output a classical representation of the trees in the ensemble. The structure of each tree consists of the set of the centroids of each internal node $\{\boldsymbol{c}_{\text{node}} \in \mathbb{R}^{d}\}_{\text{node}}$, and the leaf label information for each leaf node. With this structure, the inference is executed on the classical computer. Given a sample and the set of trees from the ensemble, the inference consists of the sample traversing through each tree individually until it reaches the leaf node. The sample goes through the tree by calculating the distance between the sample and the centroids of the children of the assigned node in the previous depth. The definition of distance used is the weighted Euclidean distance defined in the training stage using the feature weights. During training (or retraining) of the RF, we can store either the centroids or the weighted centroids (Step 1 (b)). With the weighted centroids, it can be used to perform the weighted distance calculation. Moreover, as the feature weights are stored in the classical memory, these could be used to define the weighted distance too. 

This procedure for all the trees is repeated either in sequential order or in parallel. Therefore, the time complexity to perform the inference for one sample utilizing one tree of the RF model is $T_{\text{eval}} = \mathcal{O}(kDd)$, where $k$ is the number of clusters, $D$, the maximum tree depth, and $d$, the number of features. 

\section{Algorithmic Contributions \label{algorithmic_ingredients}}

Now that we have introduced the proposed algorithm and discussed its complexity, in this section, we present the key novel algorithmic ingredients introduced in this work, part of QC-Forest-BUILD and QC-Forest-RETRAIN. 

\subsection{Feature Weight Calculation} \label{sec:feature_weight}

Feature weights quantify the statistical relationship between input feature vectors and the target variable. In this work, we consider two feature weight methods: Pearson correlation for regression and binary classification \cite{pearson1895vii}, and $\eta$ coefficient \cite{levine2002eta} for multi-class classification. \\

\noindent \textbf{1. Pearson Correlation}. Given a data set with continuous features with continuous labels (regression task), the Pearson correlation coefficient provides a method to estimate the bivariate linear relationship between each feature vector in the data and the target label. This same measure can be used to quantify the linear relationship between continuous features and a binary target label vector. In this setting, it takes the name of point-biserial. 

For constructing a given $\text{tree}_n$ in the ensemble, we sample a constant fraction subset $\mathcal{S}_n \subset D_N$. Here for simplicity and without the lack of generality, we show how to compute the feature weights for the entire dataset instead of the subset. Given the dataset $D_N$, let us refer to the $j$-th column vector as the feature vector $\boldsymbol{x}^j = [x_{1j},\cdots,x_{Nj}]^T \in \mathbb{R}^N$. Similarly, we denote the target label vector as $Y = [y_1,\cdots,y_N]^T$. The Pearson correlation coefficient between the two vectors is defined as,
\begin{equation}
    w_{j} =  \frac{\sum_{i=1}^N(x_{ij} - \mu_j)(y_i - \mu_y)}{\sqrt{\sum_{i=1}^N(x_{ij} - \mu_j)^2}\sqrt{\sum_{i=1}^N(y_{i} - \mu_y)^2}}
    \label{Eq:pearson}
\end{equation}
where $\mu_j = \frac{1}{N} \sum_{i=1}^N x_{ij}$ is the mean over the feature vector and $\mu_y = \frac{1}{N} \sum_{i=1}^N y_{i}$ is the mean over the label vector. 

We compute $\boldsymbol{w}(N) = \{w_j\}_{j\in [d]}$ and computing a single $w_j$ requires the calculation of the inner product between the two vectors to obtain the quantity $\sum_i x_{ij}y_i$, the mean of each vector $\mu_j, \mu_y$, and the inner product of the vector with itself $\sum_i x_{ij}^2, \sum_i y_i^2$. All these operations take time linear in $N$. Thus, computing the feature weights in QC-Forest-BUILD takes time $\mathcal{O}(Nd)$.

Next, the following theorem allows us to efficiently update the Pearson correlations to obtain their values corresponding to data of size $N + N_{\text{new}}$. We leverage the fact that we have calculated these correlations for $N$ samples. 

\begin{theorem} [Time complexity to classicalyl update Pearson correlation] \label{thm:pearson_correlation}
Given access to $\mu_j$, $\mu_y$ and $\sum_{i=1}^{N} x_{ij}y_{i}$ $\forall$ $j \in [d]$, calculated as part of the calculation of $\boldsymbol{w}(N) := \{w_{1}(N), \dots, w_{d}(N)\}$ for $D_{N}=\{(\boldsymbol{x_n}, y_n)\}_{n=1}^{N} \in \mathbb{R}^{d}$ x $\{l_1, \dots, l_c\}$ and access to the new batch of data with $N_{\text{new}}$ examples $D_{N_{\text{new}}}$, $\boldsymbol{w}(N+N_{\text{new}})$ is computed in time  $\mathcal{O}(N_{\text{new}} d)$.
\end{theorem}
\begin{proof}

Note that the numerator and denominator can be calculated as, 
\begin{equation}
    Num = \sum_{i=1}^{N} x_{ij}y_{i}  - 2 \mu_{j} \mu_{y}N + N \mu_{j}^{2}
\end{equation}
\begin{equation}
    Den = \sqrt{SS_{T}^{j}} \sqrt{SS_{y}}
\end{equation}
Where $SS_{T}^{j} =  \sum_{i=1}^{N} (x_{ij} - \mu_j)^2$ is the sum of the squared deviations of the feature vector and $SS_{y} =  \sum_{i=1}^{N} (y_{i} - \mu_y)^2$, the sum of the squared deviations of the labels. 

Given a new batch of data $N_{\text{new}}$, we do the following calculations,

\begin{enumerate}
    \item Given the stored $\mu_j$ (previously computed with $N$ examples), the mean for $N + N_{\text{new}}$ examples, named $\mu_j^{\text{tot}}$, can be computed as: $\mu_j^{\text{tot}} = \frac{N\mu_j + N_{\text{new}}\mu'_j}{N + N_{\text{new}}}$, where $\mu'_j$ is the mean of $N_{\text{new}}$ examples. It only requires the calculation of $\mu'_j$ and simple arithmetic calculations. This operation takes time $\mathcal{O}(N_{\text{new}}d)$ for all the features. 
    \item Similarly, we update the mean over the labels $\mu_y^{\text{tot}}$ in time $\mathcal{O}(N_{\text{new}}d)$
    \item In order to update the numerator note that, 
    \begin{equation}
        Num^{\text{tot}} = \sum_{i=1}^{N} x_{ij}y_{i} +\sum_{i=N+1}^{N_{\text{new}}} x_{ij}y_{i} -2 \mu_{j}^{\text{tot}} \mu_{y}^{\text{tot}}(N+N_{\text{new}}) + (N+N_{\text{new}}) (\mu_{j}^{\text{tot}})^{2}
    \end{equation}
    The only term needed to be calculated is the second term and that takes time, for all the features, $\mathcal{O}(N_{\text{new}}d)$\\ 
    \item The update of the denominator requires to update $SS_{T}^{j}$ and $SS_{y}^{j}$ as we see below, 

    \begin{equation}
        den^{\text{tot}} = \sqrt{SS_{T}^{(j,tot)}} \sqrt{SS_{y}^{(j,tot)}}
    \end{equation}
    
    $SS_{T}^{j}$ is updated in the same way as Step 4 in Theorem \ref{thm:eta_coef} in time $\mathcal{O}(N_{\text{new}})$. $SS_{y}^{j}$ is updated similarly, also taking $\mathcal{O}(N_{\text{new}})$ time to calculate $\sum_{i=N+1}^{N_{\text{new}}} y_{i}^2$. For all the features, the update of the denominators takes time $\mathcal{O}(N_{\text{new}}d)$. \\    
\end{enumerate}
The total time for the update of the Pearson correlations for all the features is $\mathcal{O}(N_{\text{new}}d)$ and this concludes the proof. 
\end{proof}

A key point to note is that the time complexity of updating the feature weights only has a dependency on $N_{\text{new}}$, instead of $N$. Given $N_{\text{new}} \ll N$, this procedure is highly efficient. \\


\noindent \textbf{2. Eta coefficient, $\eta$}. This coefficient, also known as the correlation ratio, is a measure of the curvilinear relationship between the statistical dispersion within individual categories and the dispersion across the whole population or sample. It measures the non-linear association between the feature vector and the label vector, i.e., the variance percentage in the label vector that can be explained by the feature vector. A key benefit of this measure is that it can handle multi-class classification data with categorical features.

Given the dataset $D_{N}$, the $\eta$ coefficient for the $j$-th feature is defined as below, 
\begin{equation}
    w_j = \eta_{j}^{2} = \frac{SS_{c}^{j}}{SS_{T}^{j}}
\end{equation}

where $SS_{T}^{j} =  \sum_{i=1}^{N} (x_{ij} - \mu_j)^2$ is the sum of the squared deviations for the vector
of the $j$-th feature, $SS_{c}^{j} = \sum_{l=l_1}^{l_c} N_{jl}  (\mu_{jl} - \mu_j)^2$ is the sum of the squared deviations for the vector of the $j$-th feature among the categories (or classes), $\mu_j$ is the mean of the $j$-th feature vector, and $\mu_{jl}$ is the mean of the $j$-th feature vector among the $N_{jl}$ examples that correspond to the label $l$ and $SS_{c}^{j}$. The complexity of computing $\boldsymbol{w}(N) = \{w_j\}_{j \in [d]}$ from scratch is $\mathcal{O}(Nd|\mathcal{M}|)$.

Similarly to what we do with Pearson correlation, we provide a theorem to efficiently update the $\eta$ coefficient for data of size $N + N_{\text{new}}$ given access to the precomputed values corresponding to the $N$ initial samples.
\begin{theorem} [Time complexity to classically update $\eta$ coefficient]\label{thm:eta_coef} 
Given access to $\mu_{j}$,  $\mu_{jl}$, $N_{jl}$, $SS_c^{(j)}$ and $SS_T^{(j)}$ $\forall$ $j \in [d], l \in [\mathcal{M}]$, calculated as a part $\boldsymbol{w}(N)$ for initial dataset $D_{N}=\{(\boldsymbol{x}_i, y_i)\}_{i=1}^{N}$ and access to the new batch of data with $N_{\text{new}}$ examples $D_{N_{\text{new}}}$, $\boldsymbol{w}(N+N_{\text{new}})$ can be computed in time  $\mathcal{O}(N_{\text{new}} d |\mathcal{M}|)$.
\end{theorem}

\begin{proof}
Similarly, as we do with Pearson correlation, we leverage the linearity of the magnitudes that are computed during the calculation of the $\eta$ coefficient. We provide the formal proof of the above theorem in Appendix~\ref{appendix:eta_coeff}.
\end{proof}

In the following Table \ref{table:feature_weights}, we summarize the time complexity and memory requirement to update the feature weight with the two proposed methods for each tree in the ensemble.
\begin{table}[h!]
\caption{Summary of the time and memory complexity of the used methods for feature weight calculation}
\label{table:feature_weights}
\begin{center}
\begin{small}
\begin{sc}
\begin{tabular}{lcccr}
\toprule
Method & Time & Memory \\
\midrule
$\eta$ coefficient    & $\mathcal{O}(N_{\textit{new}} d |\mathcal{M}|)$   & $\mathcal{O}(d |\mathcal{M}|)$  \\
Pearson    & $\mathcal{O}(N_{\textit{new}} d)$   & $\mathcal{O}(d)$  \\
\bottomrule
\end{tabular}
\end{sc}
\end{small}
\end{center}
\vskip -0.1in
\end{table}

\subsection{Supervised $q$-means and tree growth} \label{sec:weighted_clustering}

The supervised $q$-means procedure introduced by \cite{kumar2023desq} is based on the original $q$-means clustering \cite{kerenidis2019q} and it incorporates the weighted distance. In this work, we leverage this classical-quantum clustering using the classically-computed feature weights $\boldsymbol{w}(N)$ or $\boldsymbol{w}(N + N_{\text{new}})$, depending on whether we are constructing or retraining the ensemble of trees. We start at the root and upon sequential repetition of the clustering procedure grow the tree up to reach a maximum given depth $D$ with $k^D$ leaf nodes. At high level, given some initial cluster centroids $\boldsymbol{c}_1, \cdots, \boldsymbol{c}_k$, the steps are the following and they are repeated $K$ times or until convergence:

\noindent 1. Classically calculate the \textit{weighted centroids} by performing the element-wise multiplication between the feature weights and the centroids. \\
\noindent 2. Quantum subroutines are initially employed to assign each training example $\boldsymbol{x}_i$ to their nearest weighted centroids using Euclidean distance estimation, which is a practical way of implementing the weighted distance calculation with the centroids $\boldsymbol{c}_{l}$, which is defined as: $\|\boldsymbol{x}_i - \boldsymbol{c}_{l}\|_w = \sqrt{\sum_{j=1}^d w_j \cdot (x_{ij} - c_{lj})^2}$. \\
\noindent 3. Quantum centroid update subroutine calculates the new centroid values. 

This procedure as discussed in Section 5 of Des-q \cite{kumar2023desq} results in the following theorem,
\begin{theorem} 
    Given quantum access to the dataset $D_N$ for QC Forest-BUILD or $D_{N + N_{\text{new}}}$ for QC Forest-RETRAIN in time $T = \mathcal{O}(\text{poly} \log(Nd))$, the supervised $q$-means algorithms takes $K$ iterations to output with high probability the centroids $c_1, \cdots, c_k$ that are arbitrarily close to the output centroids of the $\delta$-$k$-means algorithm, a robust version of $k$-means, in time complexity,
    \begin{equation}
    \mathcal{O}\left(poly \log(Nd) \frac{K k^3d \log (k)\log(p)\log(1/\Delta)\eta_1\eta_2}{\epsilon_1\epsilon_2}\right)
    \end{equation}
    where $\eta_1 = \text{max}_i(\|\boldsymbol{x}_i\|^2)$, $\eta_2 = \text{max}_{l \in [d]}\|\boldsymbol{x}^{l}\|$, $p$ is the number of qubits to encode the distance between the weighted centroids and training examples, and $\epsilon_1, \epsilon_2, \Delta > 0$. 
\end{theorem}

Further, repeated application of this procedure creates a depth-$D$ tree in time,
\begin{equation}
    \mathcal{O}\left(\emph{poly} \log(Nd) \frac{K D k^{3D}d \log^2 (k)\log^2(p)\log^2(1/\Delta)\eta_1^2\eta_2}{\epsilon_1^2\epsilon_2}\right)
    \label{eq:depth0T}
\end{equation}

\subsection{Leaf Label Assignment} \label{sec:leaf_assign}

Once each tree in the random forest has been constructed, or retrained, each tree has at most $k^D$ leaf nodes. Each of these leaf nodes will have assigned leaf labels. The $j$-th leaf node for a given tree corresponds to the cluster $\mathcal{C}_{j, D}$. For the task of the regression, its label value is simply the mean of the label values of the samples in the cluster as shown below,
\begin{equation}
    \emph{label}_j = \frac{1}{|\mathcal{C}_{j, D}|} \sum_{i \in \mathcal{C}_{j,D}} y_i
    \label{eq:mean-label}
\end{equation}


For the task of classification with label set $\mathcal{M} \in \{l_1,\cdots,l_{|\mathcal{M}|}\}$, we propose to output the relative occurrence of each class among the samples in the cluster, 
\begin{equation}
\emph{label}_j = \{p_{jl}\}_{l \in [\mathcal{M}]} 
\end{equation}

where,
\begin{equation}
   p_{jl} = \frac{N_{jl}}{N_{j}}
    \label{eq:mean-label-class}
\end{equation}

where $N_{jl}$ is the number of samples in the $m$-th cluster with label $l_{l}$ and $N_{j}$ is the total number of samples in that leaf node, i.e., $N_{j} = |\mathcal{C}_{j, D}|$. \\ 

The method for computing labels for regression task has already been introduced by Des-q \cite{kumar2023desq}. Our main contribution is the introduction of the method for multi-class classification. We show how to efficiently get a classical estimate of the labels for multi-class classification task in the theorem below whose proof is provided in Appendix~\ref{appendix:leaf_assign}.

Note that having access to these probabilities for each tree in the ensemble allows for combining the results corresponding to each tree using different techniques. How the output of each tree is aggregated impacts the accuracy of the RF model. Refer to Appendix \ref{appendix:decay} for a discussion on two different methods for aggregating the results and their impact on accuracy. Moreover, for a given leaf node, having access to the estimated probability of each class opens up the opportunity to fine-tune the threshold to assign a class to the leaf node. This is a major improvement over Des-q \cite{kumar2023desq} that assumed a threshold $t=0.5$ limited to the binary case. In principle, the threshold depends on the dataset used. If the dataset is unbalanced, taking $t=0.5$ would result in low accuracy of the model's prediction. Refer to Appendix \ref{appendix:tuning} for a comparison and discussion of different threshold values for an unbalanced dataset.


\begin{theorem}[Time complexity of leaf label assignment for classification]
    Given access to the characteristic vector states $|\xi_{jl}\rangle = \frac{1}{\sqrt{|\mathcal{C}_{jl, D}|}} \sum_{i \in \mathcal{C}_{jl,D}}\ket{i}$, $\forall j \in [k^D], l \in [\mathcal{M}]$ prepared in time $T_{\xi_{jl}}$, where $\mathcal{C}_{jl,D}$ is the cluster of j-th leaf node with label $l$, and the label superposition state $|Y\rangle$, which is be prepared in time $\mathcal{O}(\text{poly} \log(N))$, there exists a quantum algorithm to estimate occurrence value $\{\overline{N_{jl}}\}_{j\in [k^D], l \in [\mathcal{M}]}$ such that $|\overline{N_{jl}} - N_{jl}| \leq \epsilon_4$, in time 
    $$T_{\text{leaf-label}} =\mathcal{O}\left(\frac{T_{\xi_{jl}}|\mathcal{M}| k^D\eta_3}{\epsilon_4}\right)$$ 
    where $T_{\xi_{jl}} = T_{\xi}\cdot \mathcal{O}(\log(|\mathcal{C}_{j,D}|\log (|\mathcal{M}|))\log(\mathcal{M}))$,  $N_{jl}$ being the true occurrence value of each label in the cluster, and $\eta_3 = \|Y\|$. 
    \label{appthm:cent_update-class}
\end{theorem}
\begin{proof}
    It is provided in Appendix~\ref{appendix:leaf_assign}
\end{proof}

\section{Numerical simulations: benchmark accuracy of predictions \label{section:numerics}}

The evaluation of quantum algorithms for practically-relevant use cases at the industrial scale is still far from realization as the near-term quantum hardware is characterized of having small number of noisy qubits. Moreover, there is still no commercial access to quantum-accessible data structures, which is essential for demonstrations with big datasets of the proposed algorithm. However, in our work, we evaluate the classical version of \algo, named QC-Forest-c, to quantify the accuracy of the predictions. QC-Forest-c leverages the $k$-means clustering algorithm whose performance closely aligns with robust $\delta$-k-means as shown by Kerenidis \etal \cite{kerenidis2019q}. Moreover, $\delta$-k-means is a good approximation of the performance of $q$-means. Therefore, the results obtained with \algo-c are a good approximation of the performance expected for \algo~when executed in the hybrid classical-quantum fashion. 

We showcase that QC-Forest-c has competitive performance against state-of-the-art RF methods for standard benchmark datasets of multiple sizes of both the number of samples ($N$) and features ($\#$ features) and for a variety of tasks: regression and binary and multi-class classification. The classical state-of-the-art methods, such as the one implemented in scikit-learn \cite{scikit-learn}, perform both sampling with replacement and feature bagging. For these experiments, \algo-c implements sampling with replacement and not feature bagging, which consists of selecting a random subset of the features at each candidate split. The reason is that if one or a few features are very strong predictors for the output labels, these features will be selected in many of the trees in the ensemble, making them correlated. \algo~is very much different in this sense as we do clustering in order to make the split and for this, we consider all the features in the dataset. For this reason, in below experiments we compare to a baseline that is an ensemble of single trees built with the classes in sklearn.tree and the same input data that is sampled and replaced from the datasets that is also used by QC-Forest-c to build each tree. For a comparison of QC-Forest-c to the RF classes in scikit-learn refer to Appendix \ref{appendix:comparison_random_forest}.

For each dataset, we create train-test folds and for each fold we construct and evaluate an ensemble of $100$ trees of different tree depths ($D$) and different number of clusters ($k$). We report the results obtained corresponding to the lowest values of $kD$ to achieve the highest performance because the inference time is $\mathcal{O}(kDd)$. For classification, the metric used is the area under the Receiver Operating Characteristic (ROC) Curve, also known as ROC AUC, and for regression is the root mean square error (RMSE). We report the median and the standard deviation across the folds in Table \ref{table:datasets}, which also contains the characteristics of the dataset utilized. 

For all the datasets considered, there is overlap between the metrics obtained with \algo-c and the baseline. Therefore, we affirm that \algo~obtains a competitive performance to the baseline method. Note that we report results corresponding to ensembles of shallow trees. We see a decay in performance for deeper trees due to overfitting (see Appendix \ref{appendix:decay} for details).

As mentioned, one contribution of our work is that QC-Forest can support multi-class classification by incorporating an efficient technique to calculate the $\eta$ coefficient. The performances reported in Table \ref{table:datasets} correspond to utilizing this coefficient for multi-class classification. For binary classification and regression, in principle, both the $\eta$ coefficient and the Pearson correlation could be used. Refer to \ref{appendix:technical_details} to see which method is used in each case. 

For the inference we can utilize different notions of distances. We have seen that for all the datasets considered, the weighted distance significantly improves the ROC AUC values. The values reported for \algo-c in Table \ref{table:datasets} correspond to inference using the weighted distance. Refer to Appendix \ref{appendix:weighted_distance} for a discussion about how using the weights improve performance.

\begin{table}[h!]
\caption{Performance in test of \algo-c and baseline: median and the standard deviation across the folds. The maximum depth of the trees ($D$) and the maximum number of clusters (number of children of each internal node) ($k$) are reported together with the characteristics of the dataset: \emph{N}: number of samples, \emph{$\#$ features}: number of features, \emph{classes}: number of classes, and \emph{feature type} refers if the features are numerical (\emph{NUM}) and/or categorical (\emph{CAT}). Note that "Boston Housing" is the only dataset for regression.  \label{table:datasets}}
\begin{center}
\begin{small}
\begin{sc}
\begin{tabular}{lccccccccr}
\toprule
Dataset & N & $\#$ features & Class & Feature type & D & k & \algo-c & baseline   \\
\midrule
Porto Seguro \cite{porto-seguro-safe-driver-prediction} & 60000 & 38 & 2 & Num & 2 & 4 &0.612$\pm$0.003& 0.616$\pm$0.004  \\
Spambase\cite{misc_spambase_94}   & 4601 & 57 & 2 & Num & 2 & 2 &0.936$\pm$0.023& 0.932$\pm$0.022  \\
German\cite{misc_statlog_(german_credit_data)_144}     & 1000 & 47 & 2 & Num $\&$ Cat  & 3 & 3 & 0.753$\pm$0.011& 0.741$\pm$0.017 \\
Wine \cite{misc_wine_109}     & 178 & 13 & 3  & Num  & 1 & 3 & 1.0$\pm$0.015& 0.996$\pm$0.006 \\
PIMA \cite{pima-indians-diabetes}      & 768 & 8 & 2 & Num & 3 & 4 &0.821$\pm$0.043&0.829$\pm$0.042  \\
Cars \cite{misc_car_evaluation_19}     & 1728 & 6 & 4 & Cat & 3 & 2 &1.000$\pm$0.000& 1.000$\pm$0.000  \\ 
Boston housing \cite{harrison1978hedonic} & 506  & 8 & - & Num  & 2 & 5 &0.037$\pm$0.008& 0.037$\pm$0.009  \\  
\end{tabular}
\end{sc}
\end{small}
\end{center}
\vskip -0.1in
\end{table}

We assess the performance when doing retraining with incremental batches of data with QC-Forest-c using the Porto Seguro dataset for binary classification. For this, we utilize the same parameters $D$ and $k$ reported in Table \ref{table:datasets}. We randomly sample batches of data to simulate access to new training samples from the same distribution over time. We performed incremental retraining using up to $5300$ samples to train. At each incremental retrain with old and new data, we evaluate on the same test data. We assess the performance by utilizing two different sizes of the evaluation test, while preserving the same fraction of the labels in the original dataset. To show some statistics, we repeat the sampling to create the batches of data and perform the incremental retraining followed by the evaluation on test five times. We show the results in Fig. \ref{fig:retraining}, where in (a) the test set size is $250$, being this the $0.25$ of the size of the first batch of data used to train and in (b) the test set is of size $4787$. For each of these experiments, we have seen that given the size of the data set used, at some point of the continuous incremental retraining, the RF model starts overfitting. For this reason, we report up to the values displayed, corresponding to $4000$ in (a) and $5300$ in (b). 

\begin{figure*}
\begin{subfigure}[t]{0.5\textwidth}
        \centering
        \includegraphics[height=2in]{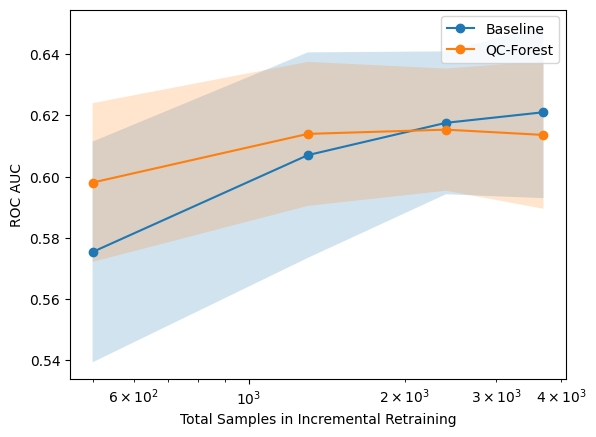}
    \end{subfigure}%
    ~ 
    \begin{subfigure}[t]{0.4\textwidth}
        \centering
        \includegraphics[height=2in]{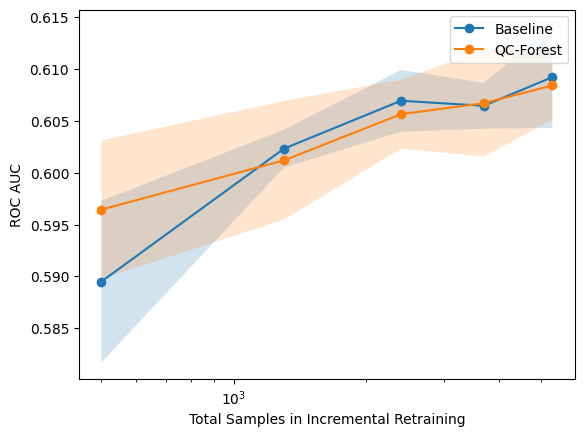}
    \end{subfigure}
\caption{ROC AUC as a function of the number of samples used to do incremental retraining. (a) corresponds to a test set of size 250 samples, corresponding to a fraction of $0.25$ of the size of the first batch used and (b) corresponds to a test size of $4787$ (total samples used in retraining are 5300). Orange line corresponds to the median value obtained with QC-Forest-c and the blue line to the one obtained with the baseline. These are the median over the five different sampling experiments and the shaded area corresponds to the standard deviation obtained. }
\label{fig:retraining}
\end{figure*}

It can be seen an improvement on the accuracy metric, the ROC AUC, as the total number of samples used in the retraining is increased, when utilizing both methods: the baseline and QC-Forest-c. What is more, it can be state that the performance obtained with QC-Forest is competitive to the baseline. However, the baseline method has a complexity linear in the number of samples, making it inefficient to handle big sets of data. In contrast, QC-Forest performs the retraining efficiently, allowing to perform the retraining at large scale and account for big number of samples. Note that how much improvement can be achieved by incremental retraining depends on the tasks and the dataset. 

On top of the comparison of the performance of QC-Forest to the baseline we did in terms of accuracy, it is interesting to make the question is the proposed method also reduces an impurity measure as the tree grows. As discussed in Section \ref{section:intro}, the state-of-the-art methods, based on threshold feature splitting, perform the split based on the idea of selecting a feature and threshold that reduce an impurity measure the most. QC-Forest does not perform the split based on this but rather by performing clustering. An interesting question is if QC-Forest manages to reduce an impurity measure as grows the tree. To answer this with numerical simulations, we utilize the Porto Seguro data set and construct the RF model with $D=5$ and $k=2$, and we calculate the total entropy at each depth of the tree. This corresponds to the sum of the entropy at each node weighted by the fraction of samples in each node. We also construct a RF model using QC-Forest that does not use feature weights, being the clustering unsupervised using the Euclidean distance. 

In Fig. \ref{fig:entropy}, we compare how the entropy is reduced as we grow on single tree using QC-Forest using $\eta$ coefficient as feature weights (supervised) and using unsupervised clustering and compare to the baseline, which is based on threshold feature splitting that minimizes the entropy. Note that the number of trees in the ensemble is one. We have seen similar behaviour for all the trees in the ensemble. We observe that for $D=2$, QC-Forest-c manages to reduce the entropy, overlapping with the baseline. For higher depths, the baseline method is much effective in reducing the entropy. The limitation in the reduction of the entropy we encounter with QC-Forest is related to the overfitting we also encounter for higher depths (refer to Appendix \ref{appendix:decay}). What is more, it can be seen that for tree depth higher than two ($D > 2$), the supervised clustering contributes to reduce the entropy more. This is also related to the improvement in the accuracy, measure by the ROC AUC, we have seen when incorporating the feature weights as we discuss in Appendix \ref{appendix:weighted_distance}.

\begin{figure*}
    \centering
    \includegraphics[height=2in]{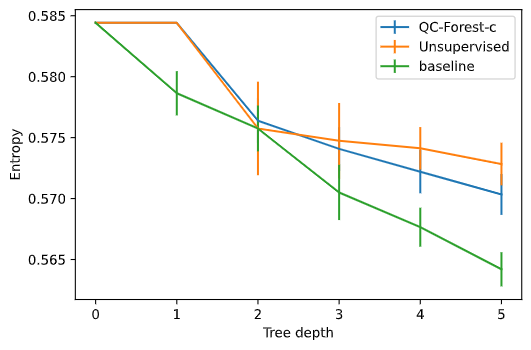}
\caption{Entropy of a single tree at each depth as a function of the depth. The entropy at each depth corresponds to the sum of the entropy of the labels at each node weighted by the fraction of training samples at each node. We compare QC-Forest-c to the baseline. ``Unsupervised'' refers to QC-Forest-c that performs unsupervised clustering that implements Euclidean distance ($w_j = 1$, $\forall j$, where $w_j$ is the weight of the $j$-th feature). This corresponds to one tree over the ensemble trained over the five folds. The lines correspond to the median values and the bars to the standard deviation across the folds. }
\label{fig:entropy}
\end{figure*}

\section{Conclusions \label{section:conclusion}}

We introduce \algo, a hybrid classical-quantum algorithm for construction and time-efficient retraining of random forest. It leverages supervised $q$-means clustering that incorporates label information into the model by using a distance definition that is weighted with feature weights, which capture the dependency of the features on the labels. \algo~achieves a significant speedup over standard methods in classical computation, which scale linear on $N$, when performing retraining in the streaming setting. In this setting data accumulates in small, and periodic increments and \algo~scales poly-logarithmically on the total number of samples ($N$)  when performing retraining with the accumulated data, i.e., incremental retraining. To the best of our knowledge, this is the first time that such time-efficient incremental retraining has been proposed for RF models. 

\algo~successfully tackles the limitations of the hybrid classical-quantum Des-q \cite{kumar2023desq} algorithm for construction and efficient retraining of single trees by: (1) extending to multi-class classification by introducing an efficient classical method to exactly calculate the $\eta$ coefficient to be used as a feature weights for multi-class classification and by also proposing an efficient quantum method to estimate the probability of each label class from the trees' leaf nodes and (2) introducing an efficient classical method to exactly calculate feature weights. What is more, Des-q has shown advantage over classical methods for single tree models for binary classification and regression. However, state-of-the-art ML models for those tasks are in the setting of ensemble of trees. \algo~ is used in the setting, paving the way to the development of quantum algorithms that improve over classical methods that construct state-of-the-art ML models to tackle the same tasks.

We have showcased that \algo~has competitive performance in terms of the accuracy of the predictions against state-of-the-art RF methods for widely used benchmark datasets containing different number of samples, up to 80,000, and number of features and for a variety of tasks: regression, and binary and multi-class classification. These benchmarks show that when the quantum hardware matures allowing for implementation and adoption of \algo, the its produced models will achieve competitive performance to the standard classical methods, while significantly speeding up the model retrain in the streaming setting, enabling the train with large data that accounts for long periods of time.  

 

\section*{Acknowledgments}
The authors thank our colleagues at the Global Technology Applied Research Center of JPMorgan Chase for helpful comments and discussions, in particular to Rudy Raymond, Rajagopal Ganesan and Alex Buts for their contributions to the numerical simulations.  

\section*{Data Availability}
The data presented in the paper is available at \url{https://doi.org/10.5281/zenodo.12610367}.

\section*{Disclaimer}

This paper was prepared for informational purposes by the Global Technology Applied Research center of JPMorgan Chase $\&$ Co. This paper is not a product of the Research Department of JPMorgan Chase $\&$ Co. or its affiliates. Neither JPMorgan Chase $\&$ Co. nor any of its affiliates makes any explicit or implied representation or warranty and none of them accept any liability in connection with this paper, including, without limitation, with respect to the completeness, accuracy, or reliability of the information contained herein and the potential legal, compliance, tax, or accounting effects thereof. This document is not intended as investment research or investment advice, or as a recommendation, offer, or solicitation for the purchase or sale of any security, financial instrument, financial product or service, or to be used in any way for evaluating the merits of participating in any transaction.

\newpage

\nocite{langley00}

\bibliographystyle{ieeetr}
\bibliography{refs}

\appendix


\section{Data Loading} \label{appendix:data_loading}

Given classical data $D_{N}=\{(\boldsymbol{x}_i, y_i)\}_{i=1}^{N}$ where each $(\boldsymbol{x}_i, y_i) \in \mathbb{R}^{d} \times \mathcal{M}$, with $\mathcal{M} = \{l_1, \dots, l_c\}$, the data loading procedure ensures the ability to efficiently query the data as quantum amplitude encoding states. For each row $\boldsymbol{x}_i = (x_{i1},\cdots,x_{id}) \in \mathbb{R}^d$, the quantum amplitude encoding is the $\lceil\log(d)\rceil$ qubits encoded state of the form,
\begin{equation}
    \ket{\boldsymbol{x}_i} = \frac{1}{\|\boldsymbol{x}_i\|} \sum_{j=1}^d x_{ij} \ket{j}
    \label{eq:amp_enc}
\end{equation}
Additionally, the encoding of the matrix $X = [\boldsymbol{x_1},\cdots, \boldsymbol{x_N}]^T \in \mathbb{R}^{N \times d}$ takes the following form,
\begin{equation}
    \ket{X} = \frac{1}{\|X\|_F} \sum_{i=1}^N \|\boldsymbol{x}_i\| \ket{\boldsymbol{x}_i} \ket{i}
    \label{eq:amp_enc_matrix}
\end{equation}
where $\ket{\boldsymbol{x}_i}$ is according to Equation \ref{eq:amp_enc}, and $\|X\|_F = \sqrt{\sum_{i=1}^N \|\boldsymbol{x}_i\|^2}$ is the matrix's Frobenius norm.

To have a time-efficient (poly-logarithmic in input size) method to prepare the states given in Eq~\ref{eq:amp_enc} and Eq~\ref{eq:amp_enc_matrix}, we leverage the quantum-accessible-data structure, KP-tree, proposal of Kerenidis \etal \cite{kerenidis2016quantum}. The KP-tree data structure is a quantum-accessible-data structure because it is a classical tree-like data structure stored in a quantum read-only-memory and accessed via superposition. This facilitates the creation of amplitude encoding states as highlighted in the Lemmas below. 

\begin{lemma} [KP-Tree data structure \cite{kerenidis2016quantum}] \label{Thm:KP-tree}
    Let $X \in \mathbb{R}^{N \times d}$ be a given dataset. Then there exists a classical data structure to store the rows of $X$ with the memory and time requirement to create the data structure being $T_{kp} = \mathcal{O}(Nd \log^2(Nd))$ such that, there is a quantum algorithm with access to the data structure which can perform the following operations (and also in superposition) in time $T = \mathcal{O}(\text{poly}\log (Nd))$,
    \begin{align}
        \ket{i}\ket{0} &\rightarrow \ket{i} \frac{1}{\|\boldsymbol{x}_i\|}\sum_{j=1}^d x_{ij}\ket{j} \\
        \ket{0} &\rightarrow \frac{1}{\|X\|_F} \sum_{i=1}^N \|\boldsymbol{x}_i\| \ket{i}
    \end{align}
\end{lemma}

Similar to Lemma~\ref{Thm:KP-tree}, the following two Lemmas load the columns elements of the matrix $X$ which correspond to the feature vectors, and the elements of the label vector $Y$ respectively,
\begin{lemma}[Column vectors superposition] \label{lemma:supcol}
        Let $X \in \mathbb{R}^{N \times d}$ be a given dataset. There exists the KP-tree structure to store the columns of $X$ given by feature vectors $\boldsymbol{x^j}:= (x_{1j}, \cdots, x_{Nj}), j \in [d]$ with $\mathcal{O}(Nd \log^2(Nd))$ memory and time requirement for data structure creation such that, there is a quantum algorithm with access to the data structure which can perform the following operations (and also in superposition) in time $T = \mathcal{O}(\text{poly}\log (Nd))$,
    \begin{align}
       \ket{0} \ket{j} \rightarrow  \left(\frac{1}{\|\boldsymbol{x^{j}}\|}\sum_{i=1}^N x_{ij}\ket{i} \right)\ket{j}  = \ket{\boldsymbol{x^{j}}} \ket{j}
    \end{align}
\end{lemma}

\begin{lemma}[Label values superposition]  \label{lemma:suplabel}
        Let $Y = [y_1,\cdots, y_N]^T \in \mathbb{R}^{N \times 1}$ be a given label vector. Then there exists the KP-tree structure to store the elements of $Y$ with $\mathcal{O}(N \log^2(N))$ memory and time requirements for data structure creation  such that, there is a quantum algorithm with access to the data structure which can perform the following operations in time $T = \mathcal{O}(\text{poly}\log (N))$,
    \begin{align}
       \ket{0}  \rightarrow  \frac{1}{\|Y\|}\sum_{i=1}^N y_i\ket{i} 
    \end{align}
    where $\|Y\| = \sqrt{\sum_{i=1}^N y_i^2}$.
\end{lemma}

\section{Feature Weight Update
}

\subsection{Eta $\eta$ coefficient \label{appendix:eta_coeff}}

\begin{theorem} [Time complexity to classically update $\eta$ coefficient. Theorem 3 Main Text]
Given access to $\mu_{j}$,  $\mu_{jl}$, $N_{jl}$, $SS_c^{(j)}$ and $SS_T^{(j)}$ $\forall$ $j \in [d], l \in c$, calculated as a part $\boldsymbol{w}(N)$ for initial dataset $D_{N}=\{(\boldsymbol{x}_i, y_i)\}_{i=1}^{N}$ and access to the new batch of data with $N_{\text{new}}$ examples $D_{N_{\text{new}}}$, $\boldsymbol{w}(N+N_{\text{new}})$ can be computed in time  $\mathcal{O}(N_{\text{new}} d |\mathcal{M}|)$.
\end{theorem}

The proof consists in showing how we update each of the parameters with the $N_{\text{new}}$ examples:

\begin{enumerate}

    \item Given the stored $\mu_j$ (previously computed with $N$ examples), the mean for $N + N_{\text{new}}$ examples, named $\mu_j^{\text{tot}}$, can be computed as: $\mu_j^{\text{tot}} = \frac{N\mu_j + N_{\text{new}}\mu'_j}{N + N_{\text{new}}}$, where $\mu'_j$ is the mean of $N_{\text{new}}$ examples. It only requires the calculation of $\mu'_j$ and simple arithmetic calculations.   This operation takes time $\mathcal{O}(N_{\text{new}}d)$ for all the features. 
    \item Similarly, we can update the means of each classes and denote them as $\mu_{jl}^{\text{tot}}$ for all $j \in [d], l \in [\mathcal{M}]$ in time $\mathcal{O}(N_{\text{new}}d|\mathcal{M}|)$.
    \item Next, we can update the number of examples in each category $N_{jl}^{\text{tot}} = N_{jl} + N'_{jl}$ where $N'_{jl}$ is the number of examples in category $l$ in the set $N_{\text{new}}$. Given access to the stored $N_{jl}$, this operation only requires to calculate $N'_{jl}$.
    \item In order to compute the sum of squared deviations ($SS_T$) for $N+N_{\text{new}}$ examples, we see that,
    \begin{equation}
    \begin{split}
        SS_T^{(j,\text{tot})} &= \sum_{i=1}^{N+N_{\text{new}}} (x_{ij} - \mu_j^{\text{tot}})^2 \\
        &= \sum_{i=1}^{N+N_{\text{new}}} (x_{ij}^2 +  (\mu_j^{\text{tot}})^2 - 2x_{ij}\mu_j^{\text{tot}}) \\
        &= \sum_{i=1}^{N+N_{\text{new}}} x_{ij}^2 - (N+N_{\text{new}}) (\mu_j^{\text{tot}})^2 \\
        &= \sum_{i=1}^{N} x_{ij}^2 + \sum_{i=1}^{N_{\text{new}}} x_{ij}^2 - (N+N_{\text{new}})
        (\mu_j^{\text{tot}})^2 \\
        &= SS_T^{(j)} + N (\mu_j^{\text{tot}})^2 \\
        &+ \sum_{i=1}^{N_{\text{new}}} x_{ij}^2 - (N+N_{\text{new}})
        (\mu_j^{\text{tot}})^2 \\
    \end{split}
    \end{equation}
    From the above equation, the only thing that is unknown is $\sum_{i=1}^{N_{\text{new}}} x_{ij}^2$ which can be computed in time $\mathcal{O}(N_{\text{new}})$.
    \item Similarly $ SS_M^{(j,\text{tot})}$ can be very simply calculated in time $\mathcal{O}(|\mathcal{M}|)$ given that we have calculated $\mu_{jl}^{\text{tot}}$ and $N_{jl}^{\text{tot}}$ in the previous steps. \\
\end{enumerate}
These calculations have to be repeated for all the $d$ features, then the complexity is $\mathcal{O}(N_{\text{new}} d |\mathcal{M}|)$ and this concludes the proof. 



\section{Leaf Label Assignment} \label{appendix:leaf_assign}

Once the tree construction procedure reaches the final depth $D$, it consists of $k^D$ leaf nodes. Our task now is to compute the label values for each leaf node. Each leaf node $j \in [k^D]$ consists of the cluster $\mathcal{C}_{j, D}$. For the task of the regression, the label value is simply the mean of the label values of the samples in the cluster i.e.,
\begin{equation}
    \emph{label}_j = \frac{1}{|\mathcal{C}_{j, D}|} \sum_{i \in \mathcal{C}_{j,D}} y_i
\end{equation}

For the task of classification with label set $\mathcal{M} \in \{l_1,\cdots,l_{|\mathcal{M}|}\}$, for each $j^{th}$ leaf node, the relative occurrence of each class $\{p_{jl}\}_{j \in [k^D], l \in [\mathcal{M}]}$ is, 
\begin{equation}
   p_{jl} = \frac{N_{jl}}{N_{j}}
\end{equation}

where $N_{jl}$ is the number of samples in the $j^{th}$ cluster, $\mathcal{C}_{j, D}$, with label $l_{l}$ and $N_{j}$ is the total number of samples in that leaf node. Having access to these probabilities allows us to combine the results obtained from the trees in the ensemble in different ways and utilize the one that brings about the best accuracy. 

Extracting the label information for the regression time time $\mathcal{O}(poly \log(Nd))$ is provided in the Des-q algorithm \cite{kumar2023desq}. For classification,  Des-q is limited to binary classes. Our novel contribution is how to efficiently perform the label extraction in the case of multi-class classification. For this task, each tree outputs the relative occurrence of each class in each of its leaf nodes $\{p_{jl}\}_{j \in [k^D], l \in [\mathcal{M}]}$, where $p_{jl}$ is defined in Eq~\ref{eq:mean-label-class}. The $j^{th}$ leaf node of a given tree corresponds to the cluster $\mathcal{C}_{j,D}$. The estimation of the occurrence of each class among the samples in the cluster can be done if we manage to create a superposition of the indices in the cluster corresponding to each label $l \in \mathcal{M}$, i.e., the objective is to create the state $$|\xi_{jl}\rangle = \frac{1}{\sqrt{|\mathcal{C}_{jl, D}|}} \sum_{i \in \mathcal{C}_{jl,D}}\ket{i}$$ where $\mathcal{C}_{jl,D}$ corresponds to the samples in the leaf node $j$ (in $\mathcal{C}_{j,D}$) with label $l$ where $|\mathcal{C}_{jl, D}|=N_{jl}$. Subsequently, the value $p_{jl}$ is simply obtained from the inner product between the index superposition state and label superposition state $|Y\rangle = \frac{1}{\|Y\|}\sum_{i=1}^N y_i \ket{i}$,
\begin{equation}
    I(|\xi_{jl}\rangle, |Y\rangle) = \langle \xi_{jl} | Y \rangle = \frac{1}{\sqrt{N_{jl}}\|Y\|} \sum_{i \in \mathcal{C}_{jl,D}} y_i = \frac{\sqrt{N_{jl}}}{\|Y\|} l_l
    \label{eq:label-inner-class}
\end{equation}
where we used that $\sum_{i \in \mathcal{C}_{jl,D}} y_i = l_lN_{jl}$. From this one can compute,
\begin{equation}
   N_{jl} = \frac{I^2(|\xi_{jl}\rangle, |Y\rangle) \|Y\|^2}{l_l^2}
\end{equation}

We can similarly compute the above value for each $l \in [\mathcal{M}]$ and $j \in [k^D]$ and get $p_{jl} = \frac{N_{jl}}{\sum_{l\in [\mathcal{M}]}N_{jl}}$.

\subsubsection{Creation of index superposition state}

The first step is to create the superposition over the indices of the samples in the cluster $\mathcal{C}_{jl,D}$ corresponding to the leaf node $j$ and label $l$ i.e., our objective is to first create the state,
\begin{equation}
    \ket{\xi_{jl}} = \frac{1}{\sqrt{|\mathcal{C}_{jl,D}|}}\sum_{i \in \mathcal{C}_{jl,D}}\ket{i}
\end{equation}

To create the above state, our first task is to create the index superposition state $|\xi_j\rangle = \frac{1}{\sqrt{|\mathcal{C}_{j, D}|}} \sum_{i \in \mathcal{C}_{j,D}}\ket{i}$ i.e., the superposition of the elements in the cluster $\mathcal{C}_{j, D}$. The theorem below provides the time complexity of preparing such a state whose proof and methodology can be found in Equation 33 in \cite{kumar2023desq}.

\begin{theorem}
     Given quantum access to the dataset $D_N$ in time $T = \mathcal{O}(\text{poly} \log(Nd))$, the index superposition state $|\xi_j\rangle = \frac{1}{\sqrt{|\mathcal{C}_{j, D}|}} \sum_{i \in \mathcal{C}_{j,D}}\ket{i}$ for each $j \in [k^D]$ can be created in time,
     \begin{equation}
     T_{\xi} = \mathcal{O}\left(\emph{poly} \log(Nd) \frac{ Dk^{2D}d \log (|\mathcal{M}|)\log (k)\log(p)\log(1/\Delta)\eta_1}{\epsilon_1}\right)
\end{equation}
 where $\eta_1 = \text{max}_i(\|\boldsymbol{x}_i\|^2)$, and $\epsilon_1 \geq 0$.
\end{theorem}

Next, in order to create $\ket{\xi_{jl}}$, we require to leverage information from the label vector $Y$. For this, not only do we store the target label $Y$ in superposition in Lemma  \ref{lemma:suplabel}, but also we store in a binarized array fashion that allows to be efficiently queried as shown below,

 \begin{lemma}[Binarized superposition over label elements] \label{lemma:suplabel-alternate}
         Let $Y \in \mathbb{R}^{N \times 1}$ be a given label vector. Then there exists a classical data structure to store the elements of $Y$ in a binarized form with the memory and time requirement to create the data structure being $T_{Y} = \mathcal{O}(N \log |\mathcal{M}|)$ such that, there is a quantum algorithm with access to the data structure which can perform the following unitaries in time $\mathcal{O}(\log (N\log |\mathcal{M}|))$,
     \begin{align}
        \ket{i}\ket{0}  \rightarrow  \ket{i}\ket{y_i}
        \label{eq:label_single}
     \end{align}
     and also in superposition,
     \begin{align}
        \frac{1}{\sqrt{N}}\sum_{i=1}^N\ket{i}\ket{0}  \rightarrow   \frac{1}{\sqrt{N}}\sum_{i=1}^N\ket{i}\ket{y_i} 
     \label{eq:label_super}
     \end{align}
     where $\ket{y_i} = \ket{y_{i1}\cdots y_{i\lceil\log m\rceil}}$ with $y_i = \sum_{j=1}^{\lceil\log m\rceil} y_{ij}2^{\lceil\log m\rceil - j}$ and $y_{ij} = \{0,1\}, \forall i,j$. 
 \end{lemma}

 \begin{proof}
     The proof follows from the fact that one can classically store the binary form of $y_i = y_{i1}\cdots y_{i\lceil\log |\mathcal{M}|\rceil}$ for each $i \in [N]$ in an ordered list, this requiring the memory and time to create the ordered list to be $T_Y = \mathcal{O}(N \log |\mathcal{M}|)$. Next, for each element $y_i$, one can convert it into the quantum state $\ket{y_i}$ by applying $\lceil\log |\mathcal{M}|\rceil$ quantum operations as follows,
     \begin{equation}
         \bigotimes_{j=1}^{\lceil\log |\mathcal{M}|\rceil}\ket{0} \xrightarrow[]{y_i}   \bigotimes_{j=1}^{\lceil\log |\mathcal{M}|\rceil} X^{y_{ij}} \ket{0}  = \ket{y_i}
     \end{equation}
     where $X$ is Pauli-$X$ gate. Thus one can perform the operations in Eq~\ref{eq:label_single} and Eq~\ref{eq:label_super} with time complexity $\mathcal{O}(\log (N\log |\mathcal{M}|))$. 
\end{proof}

We query the $Y$ memory structure in Lemma~\ref{lemma:suplabel-alternate} with the index superposition state $\ket{\xi_{j}}$ in time $T_Y = \mathcal{O}(\log(|\mathcal{C}_{j,D}|\log |\mathcal{M}|))$ and obtain the state,
\begin{equation}
    \ket{\xi_{j}}\otimes \ket{0} \rightarrow \frac{1}{\sqrt{|\mathcal{C}_{j,D}|}}\sum_{i \in \mathcal{C}_{j,D}}\ket{i}\ket{y_i}
    \label{eq:cluster-label-y}
\end{equation}

Since $y_i$ can only take discrete values in $\mathcal{M}$, we see that upon measuring the  $\ket{y_i}$ register a total of $\mathcal{O}(|\mathcal{M}|\log |\mathcal{M}|)$ times, we obtain, with high probability, all the states $\ket{\xi_{jl}}$. The time to create an individual state $\ket{\xi_{jl}}$ for a fixed $j$ and $l$ is $\mathcal{O}(\log |\mathcal{M}|)$. Thus, the total time to create $\ket{\xi_{jl}}$ is, 


\begin{equation}
\begin{split}
     T_{\xi_{jl}} &= T_{\xi}\cdot T_Y \cdot \mathcal{O}(\log |\mathcal{M}|)\\
     &= \mathcal{O}\left(\emph{poly} \log(Nd) \frac{ Dk^{2D}d \log(|\mathcal{C}_{j,D}|\log (|\mathcal{M}|))\log (|\mathcal{M}|)\log (k)\log(p)\log(1/\Delta)\eta_1}{\epsilon_1}\right)
\end{split}
\end{equation}

\subsubsection{Leaf label assignment}

Once we have prepared the state $\ket{\xi_{jl}}$, we can estimate the quantity $N_{jl}$ and thus $p_{jl}$ by querying the state $|Y\rangle = \frac{1}{\|Y\|}\sum_{i=1}^N y_i\ket{i}$ using Lemma~\ref{lemma:suplabel} and performing the inner product estimation. The time complexity is shown in the following theorem. This is performed for each leaf node $j \in [k^D]$ and each label $l \in [\mathcal{M}]$.

\begin{theorem} [Theorem 6 Main Text. Time complexity of leaf label assignment for classification]
    Given access to the characteristic vector states $|\xi_{jl}\rangle$ $\forall j \in [k^D], l \in [\mathcal{M}]$ where each state is prepared in time $T_{\xi_{jl}}$ and the amplitude-encoded states label superposition state $|Y\rangle$ which is be prepared in time $\mathcal{O}(\text{poly} \log(N))$, there exists a quantum algorithm to obtain occurrence value $\{\overline{N_{jl}}\}_{j\in [k^D], l \in [\mathcal{M}]}$ such that $|\overline{N_{jl}} - N_{jl}| \leq \epsilon_3$, in time, 
    $$T_{\text{leaf-label}} =\mathcal{O}\left(\frac{T_{\xi_{jl}}|\mathcal{M}| k^D\eta_3}{\epsilon_3}\right)$$ 
    where $N_{jl}$ is the true occurrence value of each label in the cluster, and $\eta_3 = \|Y\|$. 
    \label{appthm:cent_update-class}
\end{theorem}
\begin{proof}
We start with the initial state,
    \begin{equation}
        \ket{\phi_l} = \ket{j}\frac{1}{\sqrt{2}}(\ket{0} + \ket{1})\ket{0}
    \end{equation}

    We query the states $|\xi_{jl}\rangle$ and $|Y\rangle$ with the index $\ket{j}$ controlled on the second register which results in the mappings $\ket{j}\ket{0}\ket{0} \rightarrow \ket{j}\ket{0}|Y\rangle$ and $\ket{j}\ket{1}\ket{0} \rightarrow \ket{j}\ket{1}|\xi_{jl}\rangle$. Thus the state after this controlled rotation operation is given by,
    \begin{equation}
        \frac{1}{\sqrt{2}}\ket{j}\left(\ket{0}|Y\rangle + \ket{1}|\xi_{jl}\rangle\right)
    \end{equation}

    Applying Hadamard operation on the third qubit results in the state,
    \begin{equation}
        \frac{1}{2}\ket{j}\left(\ket{0}(|Y\rangle + |\xi_{jl}\rangle) + \ket{1}(|Y\rangle - |\xi_{jl}\rangle)\right)
        \label{eq:evolve_state_y}
    \end{equation}

    Now, the state $\ket{1}(|Y\rangle - |\xi_{jl}\rangle)$ can be rewritten as $|z_{lj}, 1\rangle$ (by swapping the registers), and hence Eq~\ref{eq:evolve_state_y} has the following mapping,
    \begin{equation}
        \ket{j}\ket{0}\ket{0} \rightarrow \ket{j}\left(\sqrt{p(1)}|z_{lj}, 1\rangle + \sqrt{1 - p(1)}\ket{G,0}\right)
        \label{eq:method2AE}
    \end{equation}
    where $G$ is some garbage state and $p(1)$ is the probability of obtaining outcome 1 when the third register of the state in Eq~\ref{eq:evolve_state_y} is measured i.e.,
    \begin{equation}
        p(1) = \frac{1}{2}\left(1 -  I(|Y\rangle, |\xi_{jl}\rangle)\right)
        \label{eq:prob}
    \end{equation} 

    Now it is clear that this is the form of the input for amplitude estimation \cite{brassard2002quantum} algorithm where the task is to estimate the unknown coefficient $\sqrt{p(1)}$. Applying amplitude estimation results in the output state (before measurement),
    \begin{equation}
        U: \ket{j}\ket{0} \rightarrow \ket{j}\left(\sqrt{\alpha}|\overline{p(1)}, G', 1\rangle + \sqrt{1 - \alpha}|G^{''},0\rangle\right)
    \end{equation}
    where $G', G^{''}$ are garbage registers. The above procedure requires $\mathcal{O}(1/\epsilon_3)$ iterations of the unitary $U$ (and its transpose) to produce the state such that $|\overline{p(1)} - p(1)| \leq \epsilon_3$. Measuring the above state results in the estimation of $\overline{p(1)}$ with a constant probability $\alpha \geq 8/\pi^2$. 

    From this, it becomes clear that we can also get an $\epsilon_3$ estimate on the inner product $I(|Y\rangle, |\xi_{jl}\rangle)$ with $\mathcal{O}(1/\epsilon_3)$ iterations of $U$. This results in the estimation of $N_{jl}$ with accuracy, 
    \begin{equation}
        |\overline{N_{jl}} - N_{jl}| \leq \frac{\|Y\|}{\sqrt{|\mathcal{C}_{j,D}|}} \epsilon_3 = \epsilon'_3
        \label{Eq:cov_cent}
    \end{equation} 
    Denoting $\epsilon'_3$ as $\epsilon_3$, the total time required to estimate the value $N_{jl}$ is the time to load the states $|Y\rangle$ and $|\xi_j\rangle$ and the subsequent time to perform the inner product estimation between them, 
    \begin{equation}
    \begin{split}
        T_{N_{jl}} &= \mathcal{O}\left(\frac{(T_\xi + \mathcal{O}(\emph{poly} \log(N)))\|Y\|}{\sqrt{|\mathcal{C}_{j,D}|}\epsilon_3}\right) \\
        &\approx \mathcal{O}\left(\frac{T_\xi\|Y\|}{\sqrt{|\mathcal{C}_{j,D}|}\epsilon_3}\right) \\
        &\leq \mathcal{O}\left(\frac{T_\xi \|Y\|}{\epsilon_3}\right)
    \end{split} 
    \end{equation}
    where we use the fact that the size of the clusters $|\mathcal{C}_{j,D}| \geq 1$. 

    Repeating the above procedure for the rest of the leaf nodes leads to the total time complexity,
    \begin{equation}
    \begin{split}
     T_{\text{leaf-label}} &=\mathcal{O}\left(\frac{T_\xi k^D\eta_3}{\epsilon_3}\right) \\
        &\mathcal{O}\left(\emph{poly} \log(Nd) \frac{ Dk^{3D}d \log (k)\log(p)\log(1/\Delta)\eta_1\eta_3}{\epsilon_1\epsilon_3}\right)
    \end{split}
        \label{eq:leaf_maj}
    \end{equation}
    where $\eta_3 = \|Y\|$.
\end{proof}

\section{Numerical Simulations} \label{appendix:further_numerical}

\subsection{Technical Details \label{appendix:technical_details}}

For all the considered datasets, we split the data into training and test using a ratio of $0.3$ and create five or ten train-test folds, depending on the dataset. Subsequently, each fold is used to train the decision tree model and the performance is evaluated in both train and test sets. We perform standard data normalization techniques on each fold, ensuring that each feature has a mean of 0 and a standard deviation of 1. We use the k-means implementation in the library pyclustering \cite{Novikov2019}. This implementation contains a feature where the user can pass as an argument the distance metric to use. We initialized the centroids for k-means with k-means++ technique \cite{arthur2007k}. The maximum number of iterations taken for k-means algorithm for the cluster centroids convergence was set to $1000$. However, it was seen that for the datasets considered, the convergence was achieved with much less iterations. Note that each internal node of the tree is split using $k$-means, with the same value of $k$. It may occur that the data cannot be divided into $k$ clusters because of insufficient samples in the node. 

For datasets corresponding to binary classification and regression we can utilize both Pearson correlation and $\eta$ coefficient to calculate the feature weights. We report the results corresponding to the method that improves the performance the most. We have seen that for PIMA, German and Porto Seguro datasets the best performance corresponds to using $\eta$ coefficient whereas for Spambase and the Boston housing, it was Pearson correlation. During inference we aggregate the results obtained in each tree of the ensemble. For the classification task, given a sample, it is input to each tree of the ensemble and each tree outputs the probability of each of the classes corresponding to the leaf node that the sample reaches, after going through the tree. We take the mean probability across the trees. For regression, we take the mean over the labels of the samples in the leaf node, we do this for each tree and then we take the mean across the trees in the ensemble. 

\subsection{Comparison to Random Forest with Feature Bagging \label{appendix:comparison_random_forest}}
State-of-the-art random forest also does \textit{feature bagging} to construct a collection of decision trees with controlled variance. Feature bagging consists of selecting a random subset of the features at each candidate split. The reason is that if one or a few features are very strong predictors for the output labels, these features will be selected in many of trees in the ensemble, making them correlated. \algo~is very much different in this sense as we do clustering in order to make the split and for this, we consider all the features in the dataset. We assign a weight, which measures how strong predictor is for the the output label but all of the features participate in the split. For this reason, the baseline method we used does not do feature bagging. However, it is interesting to compare the best results we obtained with \algo-c to a version of random forest that does this feature bagging. 

To do this, we calculate the minimum number of internal nodes for the trees in random forest such that it is at least equal to the total number of inner nodes in the trees constructed with \algo-c. In most of the cases this number is higher that the nodes in \algo-c because the split is binary in random forest. We utilize RandomForestClassifier class implemented in scikit-learn \cite{scikit-learn}. We set the number of estimators to be $100$ and for the criterion to do the split we use entropy in classification and the square error in regression. We train and evaluate in the different folds. Note that we pass the data corresponding to the fold and we do not train each tree in the forest with the same data as we did with the baseline we used to report the results in the main text. The results can be seen in Table \ref{table:comparison_random_forest}. For all the datasets considered, we see overlap between the performance of \algo-c and random forest. Therefore, we conclude that the proposed method is competitive to this baseline too. For some datasets there is more overlap than others. Note that we do not do pruning with \algo-c. We believe that incorporating this may improve the performance even more and increase the overlap with this version of random forest.

\begin{table}[h!]
\caption{Performance in test. Comparison of the best results obtained with \algo-c and the performance of random forest implemented in scikit-learn. The total number of inner nodes in the trees are reported in each case. The metrics reported are ROC AUC for classification and RMSE for regression. }
\label{table:comparison_random_forest}
\vskip 0.15in
\begin{center}
\begin{small}
\begin{sc}
\begin{tabular}{lccccr}
\toprule
Dataset & \algo-c & nodes & random forest & nodes \\
\midrule
Spambase   & 0.936$\pm$0.023  & 7 & 0.965$\pm$0.014 &  15   \\
German     & 0.753$\pm$0.011 & 40 & 0.767$\pm$0.015  &  63 \\
Cars     & 1.000$\pm$0.000 & 15 & 1.000$\pm$0.000  &  31 \\
Wine      & 1.0$\pm$ 0.015 & 4 &  1.0$\pm$0.002   &  7
\\
PIMA      &   0.821$\pm$0.043 & 85 & 0.817$\pm$0.039  & 127\\
Boston      &   0.037$\pm$0.008
 & 31 & 0.024$\pm$0.005  & 62\\ 
\bottomrule
\end{tabular}
\end{sc}
\end{small}
\end{center}
\vskip -0.1in
\end{table}

\subsection{Weighted Distance Improves Accuracy of Inference \label{appendix:weighted_distance}}

We have seen that for all the datasets considered, for a given number of clusters and depth (\textit{D}) using the weighted distance obtains either the same performance as using Euclidean distance and in some cases, it even improves it. In all the cases, the highest values of ROC AUC obtained corresponds to using the weighted distance. To show an example, we compare the performance in test when doing the inference utilizing the weighted distance with the feature weights used for training against using the Euclidean distance using the PIMA dataset in Fig. \ref{fig:performance_pima_weights}.

\begin{figure}[!h]
    \centering
    \includegraphics[width=8cm]{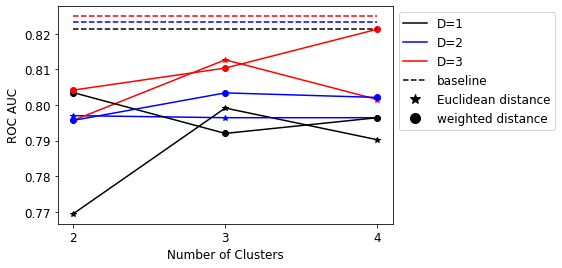}
    \caption{Performance in test. ROC AUC as a function of the number of clusters for different tree depths (\textit{D}) for an ensemble of 100 trees constructed with the PIMA dataset. }
    \label{fig:performance_pima_weights}
\end{figure}



\subsection{Threshold Tuning Improves Accuracy \label{appendix:tuning}}

For the classification task, \algo~outputs the probabilities of each class in each leaf node for trees in the ensemble. Having classical access to the probabilities allows to support multi-class classification. What is more having the probabilities enable fine tuning of the \textit{threshold} $t$. In the case of binary classification, this threshold is used such that if the fraction of samples in the leaf node with label $0$ $P(0) \geq t$, we assign the label $0$ and otherwise, we assign $1$. An example of this is the classification of credit risk as "good" or "bad". To showcase this, we use the "German" credit risk dataset \cite{misc_statlog_(german_credit_data)_144} that consists of this binary classification with a total relative accurence of label $0$ of $0.69$. We show the accuracy in test of the model constructed with \algo-c for an ensemble of $100$ trees as a function of the threshold in Fig. \ref{fig:performance_german_threshold}. It is clear that taking the naive threshold  $t=0.5$ (vertical black line), which corresponds to taking the majority vote, makes the model performance low whereas doing the fine tuning of this parameter to obtain the best threshold (vertical green line) significantly improves the accuracy.

 
\begin{figure}[!h]
    \centering
    {\includegraphics[width=8cm]{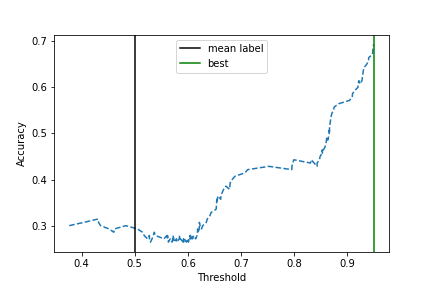}}
    \caption{Accuracy in test a function of the threshold to assign label given the probabilities of each of the two classes. The vertical black line corresponds to the threshold when taking the majority vote $t=0.5$ and the vertical green line corresponds to the "best" threshold for which the accuracy is maximum. The dataset corresponds to one of the five folds created for the "German" dataset and the model is an ensemble of 100 trees with $D=3$ and $k=3$. }
    \label{fig:performance_german_threshold}
\end{figure}

\subsection{Decay of Performance for Deep Trees and Different Methods to Aggregate Results from Trees \label{appendix:decay}}

As commented in the main text, we report the results up to small depths because for higher values we see a decrease in performance in test because of overfitting. We have seen this for all the datasets considered. We believe that \algo~is very much competitive to the baseline in terms of performance when the trees are not very much deep. As an example, we show the results obtained with the PIMA dataset in Fig.\ref{fig:pima_higher_depths}. It can be seen that the ROC AUC decreases for depths higher than three. This decrease is more notorious for higher number of clusters (\textit{k}). 

What is more, given that \algo~allows for the efficient classical estimation of the probability of each of the classes in the leaf nodes for all the trees in the ensemble, we can utilize different techniques when aggregating the results for each tree when doing inference for classification. In particular, we compare two methods. The \textit{majority} method consists in the following. Given a sample to infer, each tree in the ensemble outputs the label with highest probability in the leaf node assigned to the sample and then, it computes the probability of each class across the trees. Another method is \textit{mean} where each tree outputs the probability of each class and then it takes the mean probability of each class across the trees in the ensemble. These procedures are repeated for all the samples in the test dataset. We have seen that for all the datasets considered using the \textit{mean} method significantly improves performance. As an example, we show this comparison using the PIMA dataset in Fig. \ref{fig:pima_higher_depths}. 

\begin{figure}[!h]
    \centering
    \includegraphics[width=8cm]{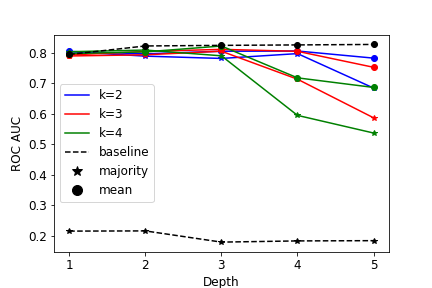}
    \caption{Performance in test. ROC AUC as a function of the tree depth for different numbers of clusters $k$ for an ensemble of $100$ trees for the PIMA dataset using $\eta$ coefficient. Given a sample to infer, \textit{majority} refers to computing the majority class in the leaf nodes assigned to the sample and then taking the distribution of probability of the assigned classes to all the samples whereas \textit{mean} refers to each tree outputting the probability of the classes in the leaf node and taking the mean probability of each of the class across the trees.}
    \label{fig:pima_higher_depths}
\end{figure}



\clearpage
\onecolumngrid

\end{document}